\documentclass[12pt]{article}
\usepackage{fullpage,amsmath}% mathtools}
\usepackage{setspace}
\usepackage{graphicx, psfrag, amsfonts}
\usepackage{natbib}
\usepackage{amsthm}
\usepackage{multirow}
\usepackage{hhline}
\usepackage[position=t,singlelinecheck=off]{subcaption}
\usepackage{caption, setspace}
\def\bth{\mbox{\boldmath $\theta$}}
\def\bbeta{\mbox{\boldmath $\beta$}}

\newcommand{\by}{\mbox{\boldmath $y$}}
\newcommand{\bz}{\mbox{\boldmath $z$}}
\newcommand{\bv}{\mbox{\boldmath $v$}}
\newcommand{\bb}{\mbox{\boldmath $b$}}
\newcommand{\bzero}{\mbox{\boldmath $0$}}

\newcommand{\mb}{\mathbf}
\newcommand{\mc}{\mathcal}

\usepackage[dvipsnames,usenames]{color}

\newtheorem{theorem}{Theorem}[section]

\newtheorem{lemma}[theorem]{\bf Lemma}
\newtheorem{corollary}[theorem]{\bf Corollary}

\newcommand{\iid}{\stackrel{iid}{\sim}}

\usepackage{color}

\graphicspath{{figures/}{figs/}}

\makeatletter
\newcommand{\labitem}[2]{%
\def\@itemlabel{\textbf{#1}{.}}
\item
\def\@currentlabel{#1}\label{#2}}
\makeatother

\title{Bayesian Restricted Likelihood Methods: Conditioning on Insufficient Statistics in Bayesian Regression}
\author{John R. Lewis, Steven N. MacEachern and  Yoonkyung Lee \\
{\small \it Department of Statistics, The Ohio State University, Columbus, Ohio 43210}\\
{\small lewis.865@buckeyemail.osu.edu, snm@stat.osu.edu and yklee@stat.osu.edu}
\thanks{This research has been supported by Nationwide Insurance Company and by the NSF under grant numbers DMS-1007682 and DMS-1209194.  The views in this paper are not necessarily those of Nationwide Insurance or the NSF.}} 
\begin{document}
\maketitle

\begin{abstract}
Bayesian methods have proven themselves to be successful across a wide
range of scientific problems and have many well-documented advantages
over competing methods. However, these methods run into difficulties
for two major and prevalent classes of problems: handling data sets
with outliers and dealing with model misspecification. We outline the
drawbacks of previous solutions to both of these problems and propose a new method as an
alternative.  When working with the new method, the data is summarized
through a set of insufficient statistics, targeting inferential quantities of interest, and the prior
distribution is updated with the summary statistics rather than the complete
data.  By careful choice of conditioning statistics, we
retain the main benefits of Bayesian methods while reducing the
sensitivity of the analysis to features of the data not captured by
the conditioning statistics. For reducing sensitivity to outliers,
classical robust estimators (e.g., M-estimators) are natural choices
for conditioning statistics. 
A major contribution of this work is the development of a data 
augmented Markov chain Monte Carlo (MCMC) algorithm
for the linear model and a large class of
summary statistics. We demonstrate the method on simulated and real data sets containing outliers and subject to model
misspecification. Success is manifested in better predictive
performance for data points of interest as compared to competing
methods.
\end{abstract}

\section{Introduction}
Bayesian methods have provided successful solutions to a wide range of scientific problems, with their value
having been demonstrated both empirically and theoretically.  Bayesian inference relies on a model consisting of three elements:  the prior distribution, the loss function, and the likelihood or sampling density.  While formal optimality of Bayesian methods is unquestioned if one accepts the validity of all three of these elements, a healthy skepticism encourages us to question each of them.  Concern about the prior distribution has been addressed through the development of techniques for subjective elicitation \citep{garthwaite2005, ohagan2006} and objective Bayesian methods \citep{berger2006}.  Concern about the loss function is reflected in, for example, the extensive literature on Bayesian hypothesis tests \citep{kass1995}.

The focus of this work is the development of techniques to handle imperfections in the likelihood $f(\by|\bth) = L(\bth|\by)$. Concern for imperfections in the likelihood are reflected in work considering minimally informative likelihoods \citep{yuan1999minimally}, sensitivities of inferences to perturbations in the model \citep{zhu2011}, the specification of a class of models and the use of Bayesian model averaging over the class \citep{clyde2004}, and considerations of such averaging when the specified class may not contain the so-called true data generating model \citep{bernardo2000, clyde2013, clarke2013Complete}.   In practice, the imperfections in a proposed likelihood often show themselves through the presence of outliers -- cases not reflecting the phenomenon under study. There are three main solutions to Bayesian outlier-handling.  The first is to replace the basic sampling density with a mixture model which includes one component for the ``good'' data and a second component for the ``bad'' data.  With this approach, the good component of the sampling density is used for prediction of future good data.  The second approach replaces the
basic sampling density with a thick-tailed density in an attempt to discount outliers, yielding techniques that 
often provide solid estimates of the center of the distribution but do not easily translate to predictive densities for further good data.  The third approach fits a flexible (typically nonparametric) model to  the data, producing a Bayesian version of a density estimate for both good and bad data.  In recent development, inference is made through the use of robust inference functions \citep{lee2014}.  

These traditional strategies  all have their drawbacks.  The outlier-generating processes 
may be transitory in nature, constantly shifting as the source of bad data changes.  This prevents us from appealing to large-sample arguments to claim that, with enough data, we can nail down a model for both good and bad data combined.  Instead of attempting to model both good and bad data, we propose a novel strategy for handling outliers. In a nutshell, we begin with a complete model  as if all of the data are good. Rather than driving the move from prior to posterior  by the full likelihood, we use only the likelihood driven by a few summary statistics which typically target inferential quantities
of interest.  We call this likelihood a restricted likelihood because conditioning is done on a restricted set of data; the set which satisfies the observed summary statistics. This restricted likelihood leads to a formal update of the prior distribution based on the sampling density of the summary statistics. 

The remainder of the paper is as follows: Section~\ref{restrictedlikelihood} introduces the Bayesian restricted likelihood and provides context with previous work, Section~\ref{illustrations} demonstrates some advantages of the methods on simple examples, and Section~\ref{BayesLinMod} details an MCMC algorithm to apply the method to Bayesian linear models. This computational strategy is a major contribution to the work, providing an approach to apply the method on realistic examples. Many of the the technical proofs are in the Appendix \ref{sec:appendix} with \texttt{R} code available from the authors. Sections \ref{simData} and \ref{RealData} illustrate the method with simulated data and a real insurance industry data set containing many outliers with a novel twist on model evaluation. A discussion (Section~\ref{Conclusions}) provides some final commentary on the new method. 

\section{Restricted Likelihood}
\label{restrictedlikelihood}

\subsection{Examples}
To describe the use of the restricted likelihood, 
we begin with a pair of simple examples for the one-sample problem.  For both, the model takes the data $\by=(y_1,\ldots,y_n)$ to be a random sample
of size $n$ from a continuous distribution indexed by a parameter
vector $\bth$, with pdf $f(y| \bth)$.  The standard, or full,
likelihood is $L(\bth | \by) = \prod_{i=1}^n f(y_i | \bth)$.  

The first example considers the case where a known subset of the data are known to be 
bad in the sense of not informing us about $\bth$.  This case mimics the setting where outliers are identified and discarded before doing a formal analysis.  Without loss of generality, we label the good cases $1$ through $n-k$ and the bad cases $n-k+1$ through $n$.  The relevant likelihood to be used to move from prior distribution to posterior distribution is clearly $L(\bth | y_1, \ldots, y_{n-k}) = \prod_{i=1}^{n-k} f(y_i | \bth)$.  For an equivalent analysis, we rewrite the full likelihood as the product of two pieces:
\begin{eqnarray}
\label{OutlyingCases}
L(\bth | \by)  
= \left( \prod_{i=1}^{n-k} f(y_i | \bth) \right) \left( \prod_{i=n-k+1}^{n} f(y_i | \bth) \right), 
\end{eqnarray}
where the second factor may not actually depend on $\bth$. We wish to keep the first factor and drop the second for better inference on $\bth$.

The second example involves deliberate censoring of small and large observations. This is
sometimes done as a precursor to the analysis of reaction time experiments  \citep[e.g.,][]{ratcliff1993} where very small and large reaction times are physiologically implausible;  explained by either anticipation or lack of attention of the subject.  
With lower and upper censoring times at $t_1$ and $t_2$, the post-censoring sampling distribution is of mixed form, with masses $F(t_1|\bth)$ at $t_1$ and $1-F(t_2|\bth)$ at $t_2$,
and density $f(y | \bth)$ for $y \in (t_1, t_2)$.  We adjust the original data $y_i$,
producing $c(y_i)$ by defining $c(y_i)= t_1$ if $y_i \leq t_1$, $c(y_i)=t_2$ 
if $y_i \geq t_2$, and $c(y_i)=y_i$ otherwise.  
The adjusted update is performed with $L(\bth |c(\by))$.  
Letting $g(t_1|\bth) = F(t_1|\bth)$,
$g(t_2 | \bth) = 1 - F(t_2|\bth)$, and $g(y|\bth)=f(y|\bth)$ for
$y \in (t_1, t_2)$, we may rewrite the full 
likelihood as the product of two pieces
\begin{eqnarray}
\label{Censoring} 
L(\bth | \by) =  \left( \prod_{i=1}^n g(c(y_i)  | \bth) \right) \left( \prod_{i=1}^n f(y_i|\bth,c(y_i)). \right) ,  
\end{eqnarray}
Only the first part is retained in the analysis. Several more examples are detailed in \cite{lewis2014}.

\subsection{Generalization}

To generalize the approach in \eqref{OutlyingCases} and
\eqref{Censoring}, %, and these other settings,
we write the full likelihood in two pieces with a conditioning statistic $T(\by)$, as indicated below:
\begin{eqnarray}
\label{FullLikelihood}
L(\bth | \by)  
& = & f(T(\by) | \bth) \,\, f(\by |\bth, T(\by)) .  
\end{eqnarray}
%Throughout the paper, we use the convention that a function is defined by its inputs. 
Here,  $f(T(\by) | \bth)$ is the conditional pdf of $T(\by)$ given $\bth$ and $f(\by |\bth, T(\by))$ is the conditional pdf of $\by$ given $\bth$ and $T(\by)$.  In the dropped case example, the conditioning statistic is $T(\by) = (y_1, \ldots, y_{n-k})$.  In 
the censoring example, the conditioning statistic is $T(\by) = (c(y_1),\ldots,c(y_n))$.  We refer to 
$f(T(\by) | \bth)$ as the restricted likelihood and $L(\bth | \by)=f(\by|\bth)$ as the full likelihood.  

Bayesian methods can make use of a restricted likelihood %in place of a complete likelihood
since $T(\by)$ is a well-defined random variable with a probability distribution indexed by $\bth$.  
This leads to the restricted likelihood posterior 
\begin{eqnarray}
\label{RestrictedPosterior}
\pi(\bth | T(\by)) & = & \frac{\pi(\bth) f(T(\by) | \bth)}{m(T(\by))} ,
\end{eqnarray}
where $m(T(\by))$ is the marginal distribution of $T(\by)$ under the prior distribution.  
Predictive statements for further (good) data rely on the model.  For another observation,
say $y_{n+1}$, we would have the predictive density 
\begin{eqnarray}
\label{RestrictedpredDist}
f(y_{n+1} | T(\by)) = \int f(y_{n+1} | \bth) \pi(\bth | T(\by))\ d\bth .  
\end{eqnarray}

\subsection{Literature review}

Our motivation for the use of summary statistics in Bayesian inference is concern about outliers or, more generally, model misspecification. Specifically, the likelihood is not specified correctly and concentrating on using well chosen parts of the data can help improve the analysis \citep[e.g.,][]{wong2004}. Direct use of restricted likelihood for this reason appears in many areas of the literature.  For example, the use of rank likelihoods is discussed by \cite{savage1969}, \cite{pettitt1983, pettitt1982}, and more recently by \cite{hoff2013}.  
\cite{lewis2012} make use of order statistics and robust estimators as choices for $T(\by)$ in the location-scale setting. 
Asymptotic properties of restricted posteriors are studied by \cite{doksum1990}, \cite{clarke1995}, \cite{yuan2004},  and \cite{hwang2005}. The tenor of these asymptotic results is that, for a variety of conditioning statistics with non-trivial regularity conditions on prior, model, and likelihood, the
posterior distribution resembles the asymptotic sampling distribution of the conditioning statistic.  

Restricted likelihoods have also been used as practical approximations to a full likelihood. For example, \cite{pratt1965} appeals to heuristic arguments regarding approximate sufficiency to justify the use of the restricted likelihood of the sample mean and standard deviation. Approximate sufficiency is also appealed to in the use of Approximate Bayesian Computation (ABC), which is related to our method.  
ABC is a collection of posterior approximation methods which has recently experienced success in applications to epidemiology, genetics, and quality control \citep[see, for example,][]{tavare1997, pritchard1999,  marjoram2003, fearnhead2012}. Interest typically lies in the full data posterior and ABC is used for computational convenience as an approximation.  Consequently, effort is made to choose an approximately sufficient $T(\by)$ and update to the ABC posterior by using the likelihood $L(\bth| \mathcal{B}(\by))$, where $\mathcal{B}(\by)=\{\by^{*}|\rho(T(\by),T(\by^{*})) \leq \epsilon\}$, $\rho$ is a metric, and $\epsilon$ is a tolerance level. This is the likelihood conditioned on the collection of data sets that result in a $T(\cdot)$ within $\epsilon$ of the observed $T(\by)$. %Hence, we can see that, as in our method, conditioning is done on a summary of the complete data. 
With an approximately sufficient $T(\cdot)$ and a small enough $\epsilon$, heuristically  $L(\bth|\mathcal{B}(\by))\approx L(\bth|T(\by))\approx L(\bth|\by)$. Consequently, the ABC posterior approximates the full data posterior and efforts have been made to formalize what is meant by  approximate sufficiency \citep[e.g.,][]{joyce2008}. ABC is related to our method in that the conditioning is on something other than the data $\by$.  However, we specifically seek to condition on an insufficient statistic to guard against misspecification in parts of the likelihood. Additionally, we develop methods where the conditioning is exact (i.e. $\epsilon = 0$).

This work extends the development of Bayesian restricted likelihood by arguing that deliberate choice of an insufficient statistic $T(\by)$ guided by targeted inference is sound practice. We also expand the class of conditioning statistics for which a formal Bayesian update can be achieved.  Our methods do not rely on asymptotic properties, nor do they rely on approximate conditioning.%, 

\section{Illustrative Examples}
\label{illustrations}
Before discussing computational details, the method is applied to two simple examples on well known data sets to demonstrate its effectiveness in situations where outliers are a major concern. The full model in each case fits into the Bayesian linear regression framework discussed in Section \ref{BayesLinMod}. %Specific details are left to the Appendix. 

The first example is an analysis of Simon Newcomb's 66 measurements of the passage time of light \citep{stigler1977}; two of which are significant outliers in the lower tail. The full model is a standard location-scale Bayesian model also used in \cite{lee2014}:
\begin{equation}
\beta\sim N(23.6, 2.04^{2}),\  \sigma^{2}\sim IG(5, 10), \ y_{i}\iid N (\beta, \sigma^{2}), i=1,2,\dots, n=66,
\end{equation}
where $y_{i}$ denotes the $i^{th}$ (recorded) measurement of the passage time of light. $\beta$ is interpreted as the passage time of light with the deviations $y_i - \beta$ representing measurement error.
Four versions of the restricted likelihood are fit with conditioning statistics: 1) Huber's M-estimator for location with Huber's `proposal 2'  for scale 2)  Tukey's M-estimator for location with Huber's `proposal 2'  for scale 3) LMS (least median squares) for location with associated estimator of scale and 4) LTS (least trimmed squares)  for location with associated estimator of scale. The tuning parameters for the M-estimators are chosen to achieve $95\%$ efficiency under normality \citep{huber2009} and, for comparability, roughly $5\%$ of the residuals are trimmed for LTS.  Two additional approaches to outlier handling are considered: 1) the normal distribution is replaced with a t-distribution and, 2) the normal distribution is replaced with a mixture of two normals. The t-model assumes $y_{i}\iid t_{\nu} (\beta, \sigma^{2})$ with $\nu=5$. The prior on $\sigma^{2}$ is $IG(5, \frac{\nu-2}{\nu}10)$ and ensures that the prior on the variance is the same as the other models. The mixture takes the form: $y_{i}\iid pN (\beta, \sigma^{2}) + (1-p)N(\beta, 10\sigma^{2})$ with the prior $p \sim \text{beta}(20,1)$ on the probability of belonging to the `good' component.

The posterior of $\beta$ under each model appears in Figure \ref{fig:newcomb_post}.  The posteriors group into two batches.  The normal model and restricted likelihood with LMS do not discount the outliers and have posteriors centered at low values of $\beta$.  These posteriors are also quite diffuse.  In contrast, the t-model, mixture model, and the other restricted likelihood methods discount the outliers and have posteriors centered at higher values.  There is modest variation among these centers.  Posteriors in this second group have less dispersion than those in the first group.   
\begin{figure}[t]
\centering
{\includegraphics[width = 4in]{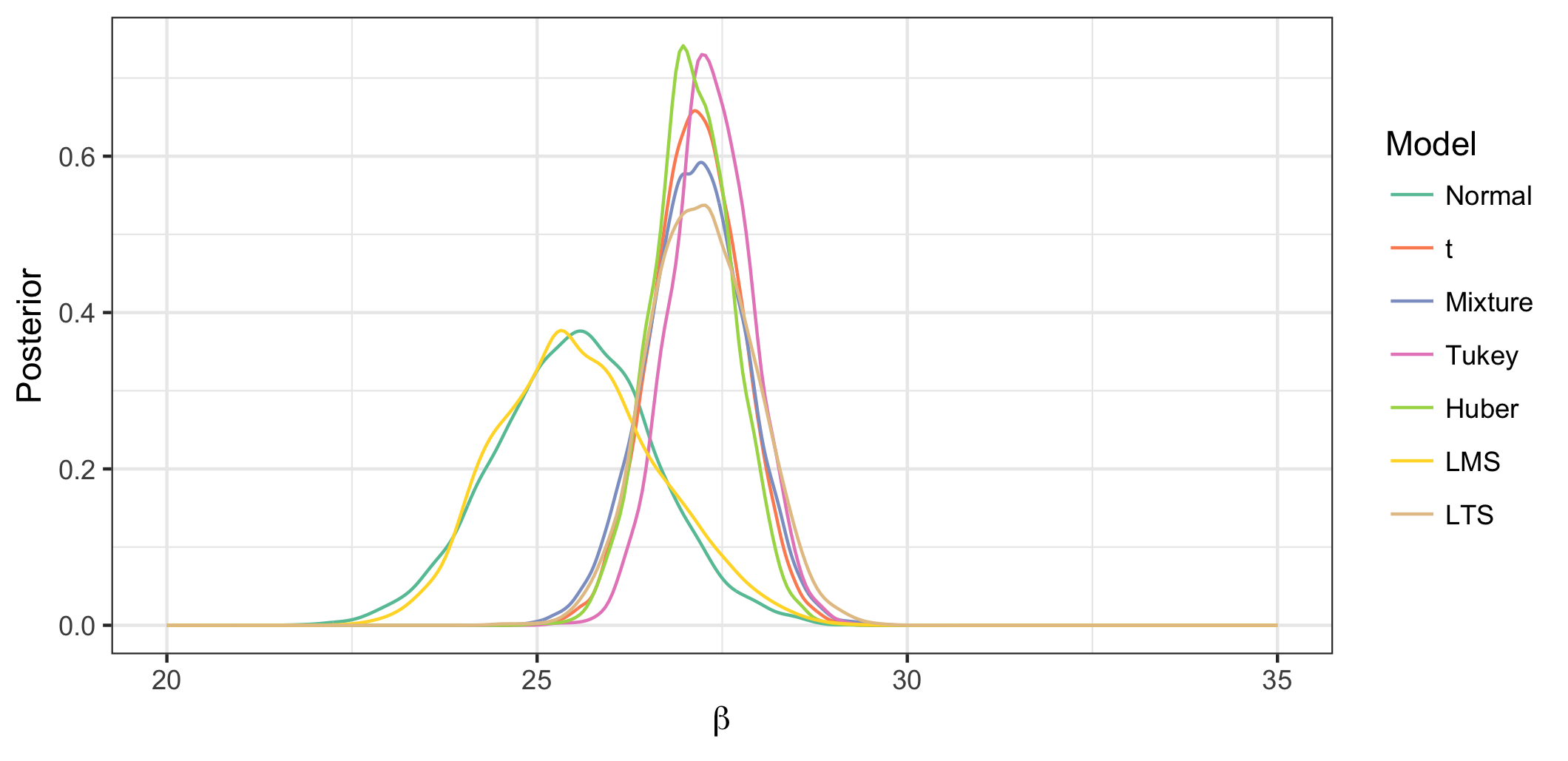}}
{\includegraphics[width = 4in]{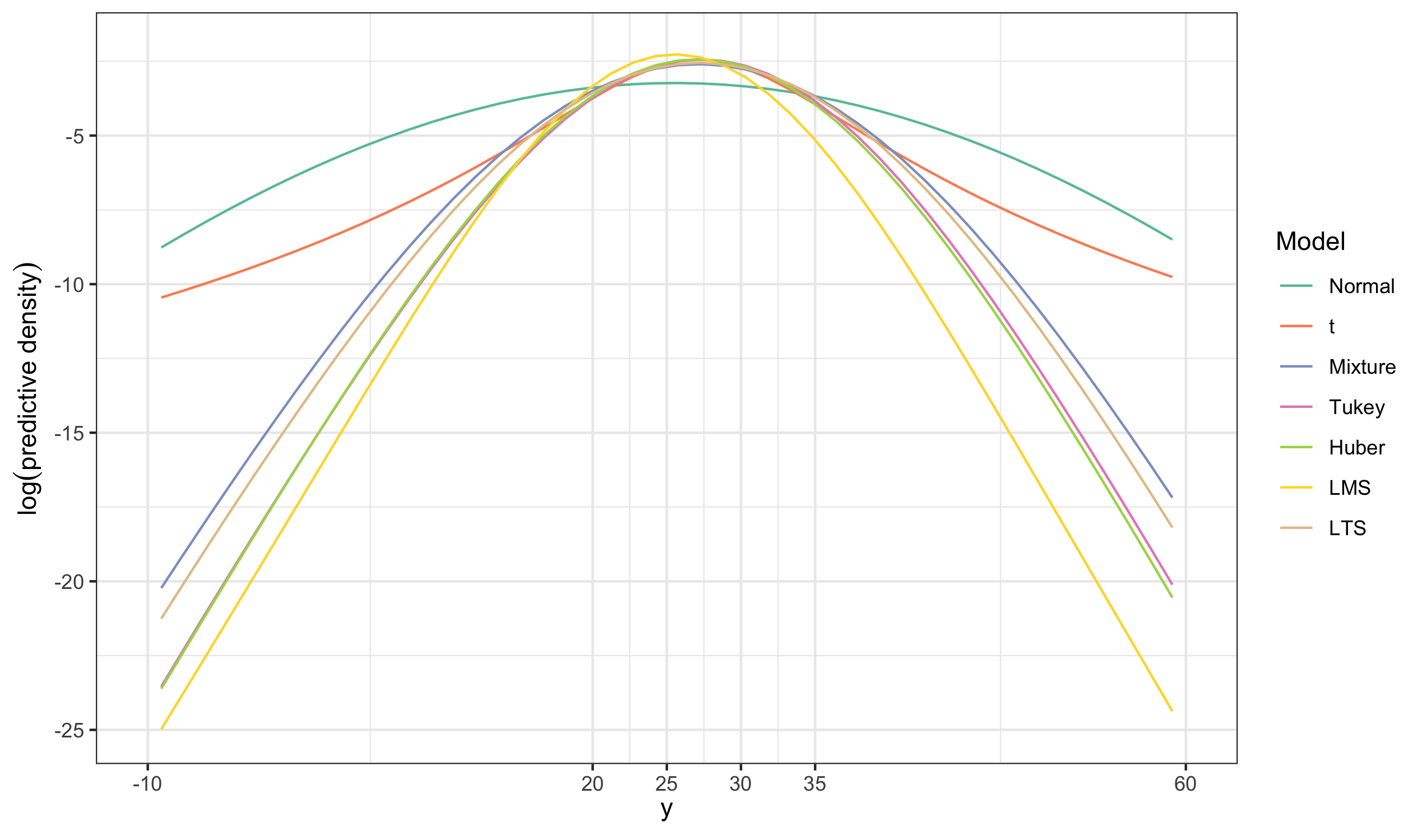}}
\caption{Results from the analysis of the speed of light data. Top: Posterior distributions of $\beta$ under each model. Bottom: Log posterior predictive distributions under each model. The differences in the tails are emphasized in the bottom plot. The horizontal axis is strategically labeled to help compare the centers of the distributions in each of the plots.}
\label{fig:newcomb_post}
\end{figure}

The pattern for predictive distributions differs (see bottom plot in Figure \ref{fig:newcomb_post}).  The normal and t-models have widely dispersed predictive distributions.  The other predictive distributions show much greater concentration.  The restricted likelihood fits based on M-estimators (Tukey's and Huber's) are centered appropriately and are concentrated. The restricted likelihood based on LTS and the mixture model results are also centered appropriately, but comparatively less concentrated. The LMS predictive is concentrated, but it is poorly centered.  

Overall, we find that the restricted likelihood methods based on M-estimators provide the most attractive analysis for these data.  They provide sharp and appropriate inference for parameters ($\beta$) and for prediction.

As a second example, a data set measuring the number of telephone calls in Belgium from 1950-1973 is analyzed. The outliers in this case are due to a change in measurement units on which calls were recorded for part of the data set. Specifically, for years 1964-1969 and parts of 1963 and 1970, the length of calls in minutes were recorded rather than the number of calls \citep{rousseeuw1987}. The full model is a standard normal Bayesian linear regression:
\begin{equation}
{\boldsymbol{\beta}}\sim N_{2}(\boldsymbol{\mu}_{0}, \boldsymbol{\Sigma}_{0}),\  \sigma^{2} \sim IG(a, b),\  \by \sim N(X\boldsymbol{\beta}, \sigma^{2} I),
\end{equation}
where $\bbeta = (\beta_{0}, \beta_{1})^{\top}$, $\by$ is the vector of the logarithm of the number of calls, and $X$ is the $n\times 2$ design matrix with a vector of 1's in the first column and the year covariate in the second.  Prior parameters are fixed via a maximum likelihood fit to the first 3 data points. In particular, the prior covariance for $\bbeta$ is set to $\Sigma_{0} = g\sigma_{0}^2 (X_{p}^{\top}X_{p})^{-1}$, with $X_{p}$ the $3\times 2$ design matrix for the first $3$ data points, $g=n=21$, $\sigma_{0} = 0.03$ and $\boldsymbol{\mu}_{0} = (1.87,  0.03)^{\top}$.  This has the spirit of a unit information prior \citep{kass1995reference} but uses a design matrix for data not used in the fit. Finally $a = 2$ and $b =1$.

Four models are compared: 1) the normal theory base model 2) a two component normal mixture model, 3) a t-model, and 4) a restricted likelihood model conditioning on Tukey's M-estimator for the slope and intercept with Huber's `proposal 2'  for scale. Each model is fit to the remaining 21 data points. The normal theory model is also fit a second time after removing observations 14-21 (years 1963 - 1970). The omitted cases consist of the obvious large outliers as well as the two smaller outliers at the beginning and end of this sequence of points caused by the change in measurement units. The mixture model allows different mean regression functions and variances for each component.  Both components have the same, relatively vague priors. The probability of belonging to the first component is given a $\text{beta}(5,1)$ prior. The heavy-tailed model fixes the degrees of freedom at 5 and uses the same prior on $\bbeta$.  The prior on $\sigma^2$ is adjusted by a scale factor of $3/5$ to provide the same prior on the variance.  

The data and  $95\%$ credible bands for the posterior predictive distribution under each model are displayed in Figure \ref{fig:calls_predictive}. The normal model fit to all cases results in a very wide posterior predictive distribution due to an inflated estimate of the variance. The t-model provides a similar predictive distribution.  The pocket of outliers from 1963 to 1970 overwhelms the natural robustness of the model and leads to wide prediction bands.  The outliers, falling toward the end of the time period, lead to a relatively high slope for the regression.  In contrast, the normal theory model fit to only the good data results in a smaller slope and narrower prediction bands.  The predictive distribution under the restricted likelihood approach is much more precise and is close to that of the normal theory fit to the non-outlying cases. The two component mixture model provides similar results, where the predictive distribution is formulated using only the good component. For these data, the large outliers are easily identified as following a distinct regression, leaving the primary component of the mixture for non-outlying data.  In a more complex situation where the outlier generating mechanism is transient (i.e., ever changing and more complex than for these data), modeling the outliers is more difficult. As in classical robust estimation, the restricted likelihood approach avoids explicitly modeling the outliers. 
\begin{figure}[t]
\centering
{\includegraphics[width = 4in]{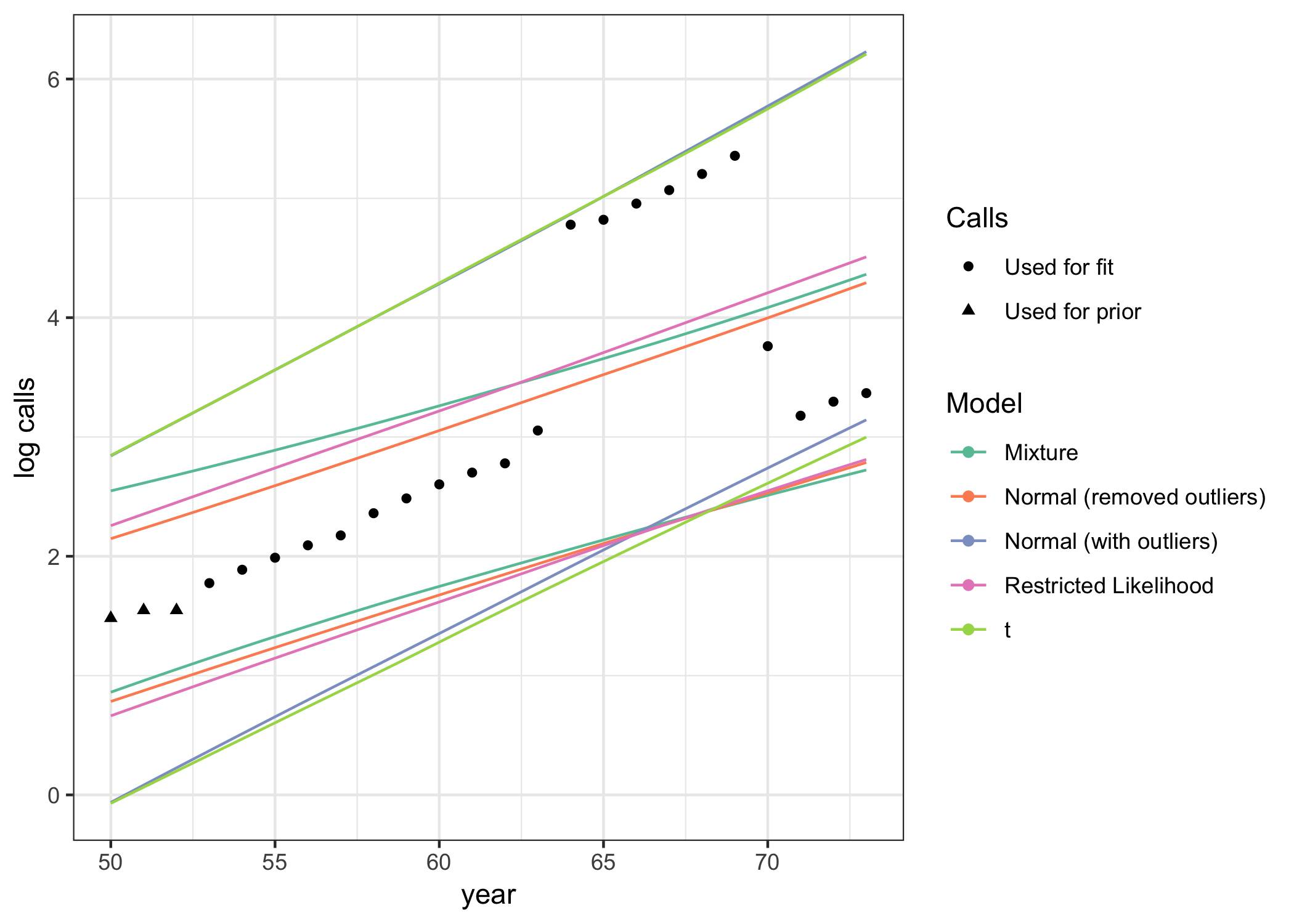}}
\caption{Pointwise posterior predictive intervals of log(calls) under the normal theory model fit to the non-outliers, the restricted likelihood model with Tukey's M-estimator for the slope and intercept with Huber's `proposal 2'  for scale, and a heavy-tailed t-distribution model. The first three data points were used to specify the prior with each model using the remaining 21 for fitting. The normal theory model was also fit after removing observations 14-20 (years 1963 - 1970).}
\label{fig:calls_predictive}
\end{figure}

\section{Restricted Likelihood for the Linear Model}
\label{BayesLinMod}

The simple examples in the previous section highlight the beneficial impact of a good choice of $T(\by)$ with the use of the restricted likelihood. This work focuses on robustness in linear models where natural choices include many used above:  M-estimators in the tradition of \cite{huber1964}, least median squares (LMS), and least trimmed squares (LTS). For these choices the restricted likelihood is not available in closed form, making computation of the restricted posterior a challenge. For low-dimensional statistics $T(\by)$ and parameters $\bth$, the direct computational strategies described in \cite{lewis2014} can be used to estimate the restricted posterior conditioned on essentially any statistic.  These strategies rely on estimation of the density of $f(T(\by)|\theta)$ using samples of $T(\by)$ for many values of $\bth$; a strategy which breaks down in higher dimensions. This section outlines a data augmented MCMC algorithm that can be applied to the Bayesian linear model when $T(\by)$ consists of estimates of the regression coefficients and scale parameter. 

\subsection{The Bayesian linear model}
We focus on the use of restricted likelihood for the Bayesian linear
model with a standard formulation: 
%\begin{eqnarray}
%\label{LinearModel}
%\bbeta & \sim & \pi_1(\bbeta); \qquad \sigma^2  \sim  \pi_2(\sigma^2)
%\nonumber\\
%y_i  & =  & x_i^\top \bbeta + \epsilon_i , \mbox{ for } i = 1, \ldots, n 
%\end{eqnarray}
\begin{eqnarray}
\label{LinearModel}
\bth&=&(\bbeta,\sigma^2) \sim  \pi(\bth) 
\nonumber\\
y_i  & =  & x_i^\top \bbeta + \epsilon_i , \mbox{ for } i = 1, \ldots, n 
\end{eqnarray}
where $x_i$ and $\bbeta \in \mathbb{R}^p$, $\sigma^2 \in \mathbb{R}^+$, 
and the $\epsilon_i$ are independent draws from a distribution with center $0$ and scale $\sigma$. $X$ denotes the design matrix whose rows are  $x_i^\top$. For the restricted likelihood model,  conditioning statistics are assumed to be of the form $T(\by) = (\bb(X, \by), s(X, \by))$ where $\bb(X, \by)= (b_1(X,\by), \dots,b_p(X,\by))^\top\in \mathbb R^{p}$ is an estimator for the regression coefficients and $s(X, \by)\in \{0\} \cup {\mathbb R}^+$ is an estimator of the scale. Throughout, observed data and summary statistic is denoted by $\by_{obs}$ and $T(\by_{obs})=(\bb(X, \by_{obs}), s(X, \by_{obs}))$, respectively. 
Several conditions are imposed on the model and statistic to ensure validity of the MCMC algorithm:
\begin{itemize}
\labitem{C1}{fullRank} The $n \times p$ design matrix, $X$, whose $i^{th}$ row is $x_i^\top$, 
is of full column rank.  
\labitem{C2}{supReal} The $\epsilon_i$ are a random sample from some distribution which has a density with 
respect to Lebesgue measure on the real line and for which the support is the real line.  
\labitem{C3}{asb}$\bb(X,\by)$ is almost surely continuous and differentiable with respect to $\by$.  
\labitem{C4}{as} $s(X,\by)$ is almost surely positive, continuous, and differentiable with respect to $\by$.  
\labitem{C5}{regEq} $\bb(X,\by+X\bv)=\bb(X,\by)+\bv \ \ \text{for  all}\ \bv\in\mathbb{R}^p$. 
\labitem{C6}{scaleEqReg} $\bb(X,a\by)=a\bb(X,\by)\ \ \ \text{for all constants } a$.  
\labitem{C7}{regIn} $s(X,\by+X\bv)=s(X,\by) \ \ \text{for all}\ \bv\in\mathbb{R}^p$.  
\labitem{C8}{scaleEq2Reg} $s(X, a\by)=|a|s(X,\by) \ \ \text{for all constants } a$.  
\end{itemize}
Properties \ref{regEq} and \ref{scaleEqReg} of $\bb$ are called
\textit{regression} and \textit{scale equivariance},
respectively.  Properties \ref{regIn} and \ref{scaleEq2Reg} of $s$ are called \textit{regression invariance}
and \textit{scale equivariance}. 
Many estimators satisfy the above properties, including simultaneous M-estimators \citep{huber2009, maronna2006} for which the R package \texttt{brlm} (\texttt{github.com/jrlewi/brlm}) is available to implement the MCMC described here. Further software development is required to extend the MCMC implementation beyond these M-estimators. The package also implements the direct computational methods described in \cite{lewis2014}. These methods are effective in lower dimensional problems and were used in both examples in Section \ref{illustrations}.

\subsection{Computational strategy}
\label{highDim}

The general style of algorithm we present is a data augmented
MCMC targeting $f(\bth, \by |
T(\by)=T(\by_{obs}))$, the joint distribution of $\bth$ and the full
data given the summary statistic $T(\by_{obs})$. 
 The Gibbs sampler \citep{gelfand1990} iteratively samples from the
 full conditionals 1) $\pi(\bth|\by, T(\by)=T(\by_{obs}))$ and 2) $f(\by|\bth, T(\by)=T(\by_{obs}))$.  When $\by$ has the summary statistic $T(\by) = T(\by_{obs})$,
the first full conditional is the same as the full data posterior $\pi(\bth|\by)$. In this case, the condition $T(\by) = T(\by_{obs})$ is redundant.  This allows us to make use of conventional MCMC steps for generation of $\bth$ from the first full conditional.  For typical regression models, algorithms abound. Details of the recommended algorithms depend on details of
the prior distribution and sampling density and we assume this can be done \citep[see e.g.,][]{liu1994, liang2008}.  

For a typical model and conditioning statistic, the second full conditional $f(\by|\bth, T(\by)=T(\by_{obs}))$ %\blue{ =f(\by|\bth, \by\in\mc A)}$ 
is not available in closed form.  We turn to Metropolis-Hastings \citep{hastings1970},
using the strategy of proposing full data $\by \in \mathcal{A}:=\{\by \in \mathbb{R}^n | T(\by)=T(\by_{obs})\}$ from a well defined distribution with support $\mathcal{A}$ and either accepting or rejecting the
proposal. Let $\by_p, \by_c \in \mathcal{A}$ represent the proposed and current
full data, respectively. Denote the proposal distribution for $\by_{p}$ by $p(\by_p|\bth,T(\by_p) = T(\by_{obs})) = p(\by_p|\bth,\by_p \in \mathcal{A}) = p(\by_p | \bth)$.  The last equality follows from the fact that our $p(\cdot | \bth)$ assigns probability one to the event $\{ \by_p \in \mathcal{A} \}$.  These equalities still hold if the dummy argument $\by_p$ is replaced with $\by_c$.  The conditional density is
\begin{eqnarray*}
f(\by | \bth, \by \in \mathcal{A}) =  \frac{f(\by | \bth) I(\by \in \mathcal{A})}{\int_\mathcal{A} f(\by | \bth) d\by} 
      = \frac{f(\by | \bth)}{\int_\mathcal{A} f(\by | \bth) d\by} 
\end{eqnarray*}
for $\by \in \mathcal{A}$ and $I(\cdot)$ the indicator function.  This includes both $\by_p$ and $\by_c$.  The Metropolis-Hastings acceptance probability  is the minimum of 1 and $R$, where
\begin{eqnarray}
\label{MHRatio}
R & = & \frac{f(\by_p|\bth,\by_p \in \mathcal{A})}{f(\by_c|\bth,\by_c \in \mathcal{A})}  
                \frac{p(\by_c|\bth, \by_c \in \mathcal{A})}{p(\by_p|\bth,\by_p \in \mathcal{A})} \\
  & = & \frac{f(\by_p | \bth)}{\int_\mathcal{A} f(\by | \bth) d\by} \frac{\int_\mathcal{A} f(\by | \bth) d\by}{f(\by_c | \bth)} \frac{p(\by_c | \bth)}{p(\by_p | \bth)} \\
 & = & \frac{f(\by_p|\bth)}{f(\by_c|\bth)} \frac{p(\by_c|\bth)}{p(\by_p|\bth)} .  
\end{eqnarray}

For the models we consider, evaluation of $f(\by | \bth)$ is straightforward.  Therefore, the difficulty in implementing this Metropolis-Hastings step manifests  itself in the ability to both simulate from and evaluate $p(\by_p | \bth)$--the well defined distribution with support $\mathcal{A}$. We now discuss such an implementation method for the linear model in \eqref{LinearModel}.

\subsubsection{Construction of the proposal}
%\textcolor{blue}{some major revisions as indicated in the responseToReferees}
Our computational strategy relies on proposing $\by$ such that $T(\by) = T(\by_{obs})$ where $T(\cdot) = (\bb(X, \cdot), s(X, \cdot))$ satisfies the conditions \ref{asb}-\ref{scaleEq2Reg}. It is not a simple matter to do this directly, but with the specified conditions, it is possible to scale and shift any $\bz^{*} \in \mathbb{R}^{n}$ which generates 
a positive scale estimate to such a $\by$ via the following Theorem, whose proof is in the appendix. 
\begin{theorem}
\label{Transformation}
Assume that conditions \ref{as}-\ref{scaleEq2Reg} hold.  Then, any vector $\bz^* \in \mathbb{R}^n$ with conditioning statistic
$T(\bz^*)$ for which $s(X,\bz^*) > 0$ can be transformed into $\by$ with conditioning statistic $T(\by) = T(\by_{obs})$ 
through the transformation 
\[
\by = h(\bz^*) := \frac{s(X,\by_{obs})}{s(X,\bz^*)} \bz^* + X\left(\bb(X,\by_{obs}) - \bb(X,\frac{s(X,\by_{obs})}{s(X,\bz^*)} \bz^*)\right) .  
\]
\end{theorem}

\noindent Using the theorem, the general idea is to first start with an initial vector $\bz^*$ drawn from a known distribution, say $p(\bz^*)$, and transform via $h(\cdot)$ to $\by \in \mathcal{A}$. The proposal density $p(\by|\bth)$ is then a change-of-variables adjustment on $p(\bz^*)$ derived from $h(\cdot)$.
In general however, the mapping $h(\cdot)$ is many-to-one: for any $\bv\in \mathbb{R}^{n}$ and any $c\in \mathbb{R}^{+}$, $c\bz^{*} + X\bv$ map to the same $\by$. This makes the change-of-variables adjustment difficult.
We handle this by first noticing that the set $\mathcal{A}$ is an $n - p - 1$ dimensional space:  there are $p$ constraints imposed by the regression coefficients and one further constraint imposed by the scale. Hence, we restrict the initial $\bz^*$ to an easily understood $n - p - 1$ dimensional space.  Specifically, this space is  the unit sphere in the orthogonal complement of the column space of the design matrix: $\mathbb{S} := \{\bz^* \in \mathcal{C}^{\perp}(X)\  |\  ||\bz^*|| = 1\}$, where $\mathcal{C}(X)$ and  $\mathcal{C}^{\perp}(X)$ are the column space of $X$ and its orthogonal complement, respectively. The mapping $h: \mathbb{S} \rightarrow \mathcal{A}$ is one-to-one and onto. A proof is provided by Theorem \ref{1to1onto} in the appendix.  The one-to-one property makes the change of variables more feasible. The onto property is important so that the support of the proposal distribution (i.e. the range of $h(\cdot)$) contains the support of the target  $f(\by | \theta, y\in \mathcal{A})$, a necessary condition for convergence of the Metroplis-Hastings algorithm (in this case the supports are both $\mathcal{A}$).

Given the one-to-one and onto mapping $h: \mathbb{S} \rightarrow \mathcal{A}$, the general proposal strategy is summarized as follows:
\begin{enumerate}
\item Sample $\bz^*$ from a distribution with known density on $\mathbb{S}$.
\item Set $\by = h(\bz^*)$ and calculate the Jacobian of this transformation in two steps.
\begin{enumerate}
\item Scale from $\mathbb{S}$ to the set $\Pi(\mathcal{A}):= \{\bz\in \mathbb{R}^n |\ \exists\ \by\in \mathcal{A}\ s.t.\ \bz=Q \by \}$ with $Q = I - XX^{\top}$. \footnote{We have used condition \ref{fullRank} to assume  without loss of generality  that the columns of $X$ form an orthonormal basis for $\mc{C}(X)$ (i.e., $X^\top X=I$).} $\Pi(\mathcal{A})$ is the projection of $\mathcal{A}$ onto $\mathcal{C}^{\perp}(X)$ and, by condition \ref{regIn}, every element of this set has $s(X, \bz) = s(X, \by_{obs})$. Specifically, set $\bz=\frac{s(X,\boldsymbol{y}_{obs})}{s(X, \boldsymbol{z}^{*})}\bz^{*}$. There are two pieces of this Jacobian: one for the scaling and one for the mapping of the sphere onto $\Pi(\mathcal{A})$. The latter piece is given in equation \eqref{cosine}.
\item Shift  from $\Pi(\mathcal{A})$ to $\mathcal{A}$: $\by=\bz+X\left(\bb(X, \by_{obs})-\bb(X, \bz)\right)$. This shift is along the column space of $X$ to the unique element in $\mathcal{A}$. The Jacobian of this transformation is given by equation \eqref{eq:volume}.
\end{enumerate}
\end{enumerate}
The final proposal distribution including the complete Jacobian is given in equation \eqref{dens:ystst} with details in the next section. Before giving these details we provide a visualization  in Figure~\ref{fig:sampSpace} of each of the sets described above using a notional example to aid in the understanding of the strategy we take. In the figure, $n = 3$, $p=1$, and the conditioning statistic is $T(\by)=(\min(\by), \sum (y_i - \min(\by))^2)$. The set $\mathcal{A}$ is depicted for $T(\by_{obs})=(0,1)$ which we describe as a ``warped triangle'' in light blue, with each side corresponding to a particular coordinate of $\by$ being the minimum value of zero. The other two coordinates are restricted by the scale statistic to lie on the quarter circle of radius one in the positive orthant. In this example, the column vector $X=\bf{1}$ (shown as a reference) spans $\mc{C}(X)$  and $\mathbb{S}$ is a unit circle on the orthogonal plane (shown in red). $\Pi(\mathcal{A})$ is depicted as the bowed triangle in dark blue. We will come back to this artificial example in the next section in an attempt to visualize the Jacobian calculations.

\begin{figure}[t]
\centering
\includegraphics[width=2.75in]{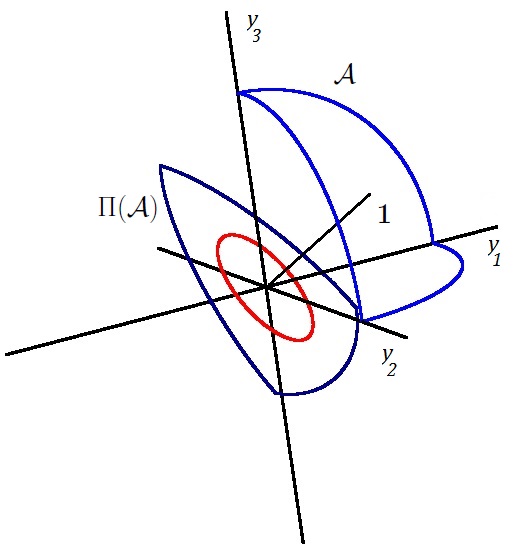}
\caption{A depiction of $\mathcal{A}$, $\Pi(\mathcal{A})$, and the
  unit circle for the illustrative example where $b_{1}(\mathbf{1},\by)=\min(\by)=0$ and
  $s(\mathbf{1},\by)=\sum (y_i -b_{1}(\mathbf{1},\by))^2 =1$.
$\mathcal{A}$ is the combination of three quarter circles, one
  on each plane defined by $y_i=0$. The projection of this manifold
  onto the deviation space is depicted by the bowed triangular shape
  in the plane defined by $\sum y_i=0$. The circle in this plane
  represents the sample space for the intermediate sample $\bz^*$. Also
  depicted is the vector $\mathbf{1}$, the design matrix for the
  location and scale setting.}
\label{fig:sampSpace}
\end{figure}

\subsubsection{Evaluation of the proposal density} 

We now explain each step in computing the Jacobian described above.
\vskip 0.05 in
\noindent
{\bf Scale from $\mathbb{S}$ to $\Pi(\mathcal{A})$} \\
The first step is constrained to $\mc{C}^\perp(X)$  and scales the initial $\bz^{*}$ to $\bz=\frac{s(X,\boldsymbol{y}_{obs})}{s(X, \boldsymbol{z}^{*})}\bz^{*}$. For the Jacobian, we consider two substeps: first, the distribution on  $\mathbb{S}$ is transformed to that along a sphere of radius $r=\|\bz\|={s(X,\boldsymbol{y}_{obs})}/{s(X, \boldsymbol{z}^{*})}$. By comparison of the volumes of these spheres, this transformation contributes a factor of $r^{-(n-p-1)}$ to the Jacobian. For the second substep, the sphere of radius $r$ is deformed onto $\Pi(\mathcal{A})$.  This deformation contributes an attenuation to the Jacobian equal to the ratio of infinitesimal volumes in the tangent spaces of the sphere and $\Pi(\mathcal{A})$ at $\bz$.  
Restricting to $\mc{C}^\perp(X)$, this ratio is the cosine of the angle between the normal 
vectors of the two sets at $\bz$.  The normal to the sphere is its radius vector $\bz$. The normal to
$\Pi(\mathcal{A})$ is given in the following lemma with proof provided in the Appendix.  Gradients denoted by $\nabla$ are with respect to the data vector.
\begin{lemma}
\label{gradSTheoremReg}
Assume that conditions \ref{fullRank}-\ref{supReal}, \ref{as}, and \ref{regIn} hold and $\by\in \mathcal{A}$. Let 
$\nabla s(X,\by)$ denote the
gradient of the scale statistic with respect to the data vector evaluated at
$\by$.  Then $\nabla s(X,\by)\in \mc{C}^\perp(X)$ and is 
normal to $\Pi(\mathcal{A})$ at $\bz=Q\by$  in $\mc{C}^\perp(X)$.
\end{lemma}
As a result of the lemma, the contribution to the Jacobian of this attenuation is 
\begin{equation}
\label{cosine}
\cos(\gamma)=\frac{\nabla s(X,\by)^\top \bz}{\|\nabla
s(X,\by)\| \|\bz\|},
\end{equation}
where $\gamma$ is the angle between the two normal vectors.
This step is visualized in Figure~\ref{fig:stretchDeform} for the notional
location-scale example.  The figure pictures only $\mathcal{C}^{\perp}(X)$,
which in this case is a plane. The unit sphere (here, the
solid circle) is stretched to the dashed sphere, contributing
$r^{-(n-p-1)}$ to the Jacobian as seen in panel (a). In panel (b), the
dashed circle is transformed onto $\Pi(\mc A)$, contributing
$\cos(\gamma)$ to the Jacobian. The normal vectors in panel (b) are
orthogonal to the tangent vectors of $\Pi(\mc A)$ and the circle. 

\begin{figure}[t]
\centering
{\includegraphics[width=2.9in]{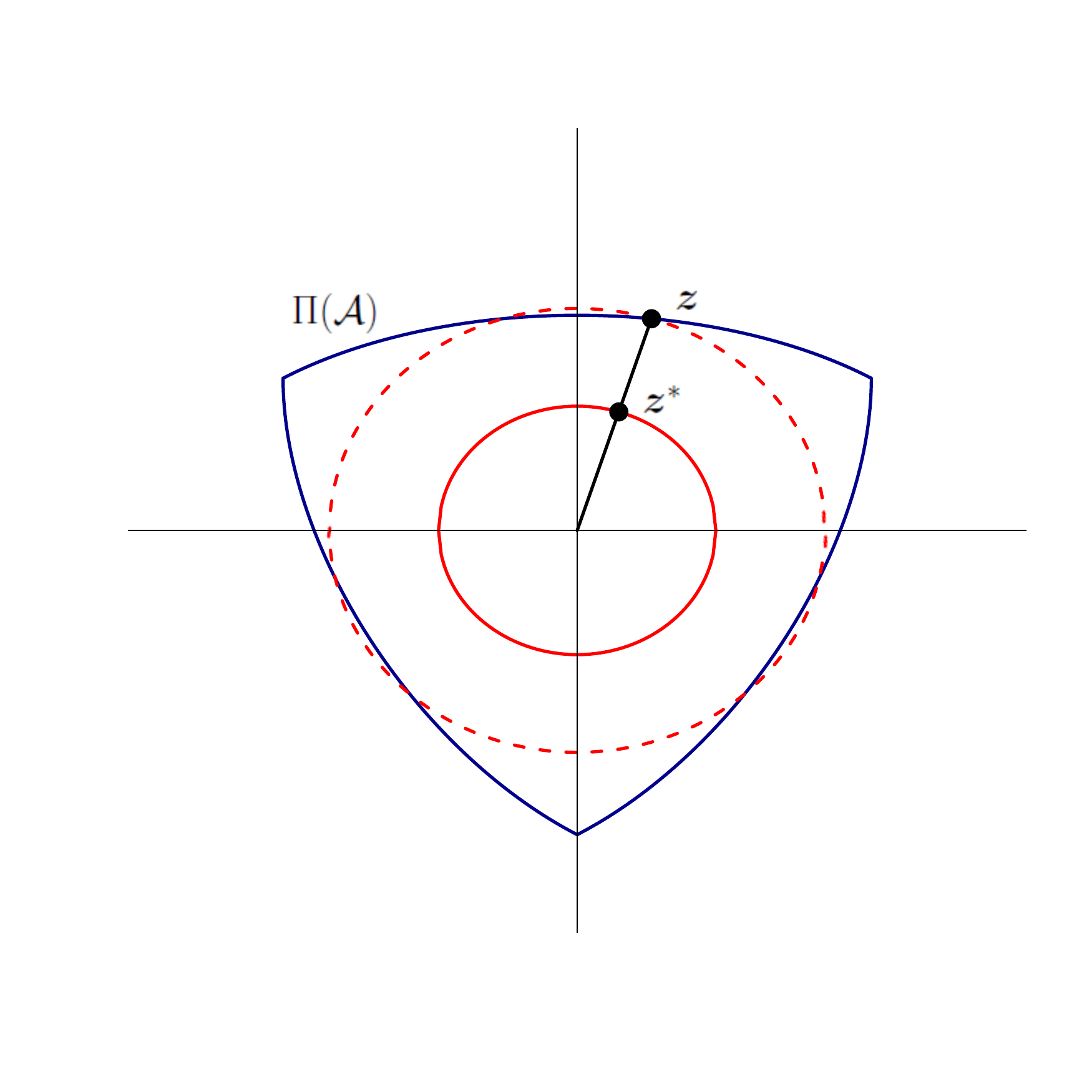}}
{\includegraphics[width=2.9in]{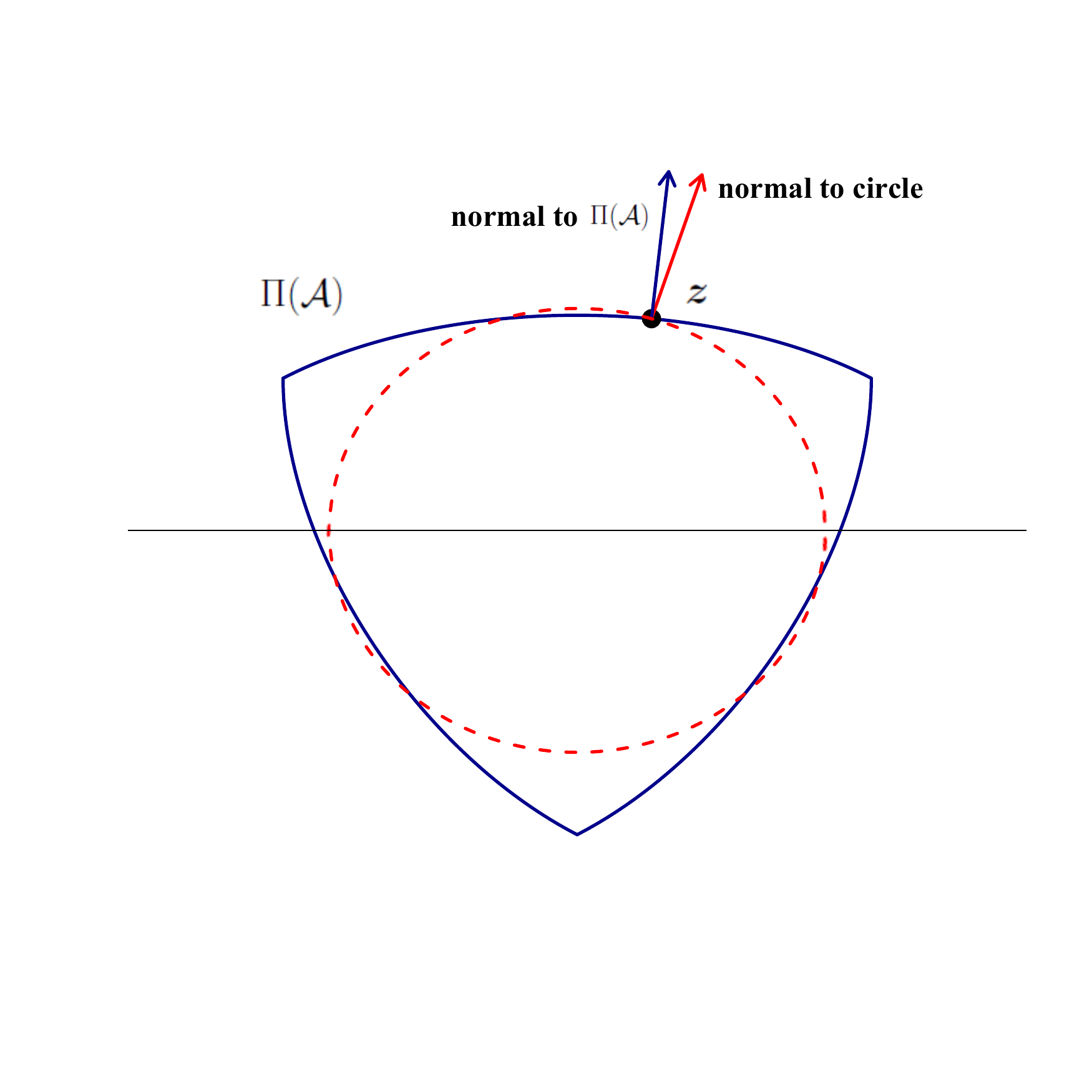}}
\caption{Visualization of the scaling from $z^{*}$ to $z$. Left: the first substep scales $z^{*}$ on the unit circle to the circle of radius $r = ||z||$, resulting in a change-of-variables transformation for the unit circle to a circle of radius $r$. The contribution to the Jacobian of this transformation is $r^{-(n-p-1)}$. Right: The second substep accounts for the the change-of-variables transformation from the circle of radius $r$ to $\Pi(\mathcal{A})$. The normal vectors to these two sets are used to calculate the contribution to the Jacobian of this part of the transformation are shown in the figure.}
\label{fig:stretchDeform}
\end{figure}

\vskip 0.05 in
\noindent
{\bf Shift from $\Pi(\mathcal{A})$ to $\mathcal{A}$} \\
The final piece of the Jacobian comes from the transformation from
$\Pi(\mathcal{A})$ to $\mathcal{A}$.  %For this we return to the full $n$ dimensional space.  
This step involves a shift of
$\bz$ to $\by$ along the column space of $X$. Since the shift depends on 
$\bz$, the density on the set 
$\Pi(\mathcal{A})$ is deformed by the shift. The
contribution of this deformation to the Jacobian is, again,
the ratio of the infinitesimal volumes along $\Pi(\mathcal{A})$ at $\bz$ to the
corresponding volume along $\mathcal{A}$ at $\by$. 
The ratio is calculated by considering the volume of the
projection of a unit hypercube in the tangent space of $\mathcal{A}$
at $\by$ onto $\mc{C}^\perp(X)$.
Computational details are
given in the following lemmas and subsequent theorem. Proofs of the lemmas are given in the appendix and the theorem is a direct result of the lemmas. Throughout, let
$\mc T_{y}(\mc A)$ and $\mc T_{y}^{\perp}(\mc A)$ denote the tangent
space to $\mc A$ at $\by$ and its orthogonal complement, respectively. %All gradients denoted by $\nabla$ are with respect to the data vector.
\begin{lemma}
\label{lem:basis}
Assume that conditions \ref{fullRank}-\ref{regEq} and \ref{regIn}-\ref{scaleEq2Reg} hold.  Then the $p+1$ gradient vectors 
$\nabla s(X,\by), \nabla b_1(X,\by),\dots, \nabla b_p(X,\by)$ form a
basis for $\mc T_{y}^\perp(\mc A)$ with probability one.
\end{lemma}

The lemma describes construction of a basis for $\mc T_{y}^\perp(\mc A)$, leading to a 
basis for $\mc T_{y}(\mc A)$.  Both of these bases can be orthonormalized.  
Let $A=[a_{1},\dots,a_{n-p-1}]$  and $B=[b_1,\dots,b_{p+1}]$ denote the 
matrices whose columns contain the orthonormal bases for  $\mc T_{y}(\mc A)$ and  $\mc T^{\perp}_{y}(\mc A)$, respectively.  
The columns in $A$ define a unit hypercube in $\mc T_{y}(\mc
  A)$ and their projections onto $\mc{C}^\perp(X)$ define a parallelepiped.
We defer construction of $A$ until later. 

\begin{lemma}
\label{lem:fullrank}
Assume that conditions \ref{fullRank}-\ref{regEq} and \ref{regIn}-\ref{scaleEq2Reg} hold.  
Then the $n\times (n-p-1)$ dimensional matrix $P=QA$ is of full column rank.
\end{lemma}

As a consequence of this lemma, 
the parallelepiped spanned by the columns of $P$ is not
degenerate (it is $n-p-1$ dimensional), and its volume
is given by
\begin{equation}
\label{eq:volume}
\text{Vol} (P) := \sqrt{\text{det}(P^\top P)}=\prod_{i=1}^{r} \sigma_i
\end{equation}
where $r=\text{rank} (P)=n-p-1$ and $\sigma_1\geq
\sigma_2\geq\dots\geq\sigma_r>0$ are the singular values of $P$ (e.g.,
\cite{miao1992}). 
Combining Lemmas \ref{lem:basis} and \ref{lem:fullrank} above leaves us with the following result concerning the calculation of the desired Jacobian.  
\begin{theorem}
\label{Jacobian}
Assume that conditions \ref{fullRank}-\ref{regEq} and \ref{regIn}-\ref{scaleEq2Reg} hold.  Then the
Jacobian of the transformation from the distribution along 
$\Pi(\mc A)$ to that along $\mc A $ is equal to the volume given in \eqref{eq:volume}.
\end{theorem}

\vskip 0.05 in
\noindent
{\bf The proposal density} \\
Putting all the pieces of the Jacobian together we have the following result. Any dependence on other variables, including current states in the Markov chain, is made implicit. 
\begin{theorem} 
Assume that conditions \ref{fullRank}-\ref{scaleEq2Reg} hold.  Let $\bz^{*}$ be sampled on the unit sphere in $\mc {C}^\perp (X)$ with density $p(\bz^{*})$.  Using the transformation of $\bz^*$ to $\by\in \mc A$ described in Theorem \ref{Transformation}, the density of $\by$ is
\begin{equation}
\label{dens:ystst}
p(\by)=p(\bz^*) r^{-(n-p-1)} \cos(\gamma)\text{Vol} (P)
\end{equation}
where $r={s(X,\boldsymbol{y}_{obs})}/{s(X,  \boldsymbol{z}^{*})}$,
and $\cos(\gamma)$ and $\text{Vol} (P)$ are as in equations \eqref{cosine} and \eqref{eq:volume}, respectively. 
\end{theorem} 

%In practice, computing $A$ directly to find $P$ and $\text{Vol} (P)$ is computationally intensive as it involves orthogonalization of $n$ vectors in $n$-dimensional space. 
Some details for computing the needed quantities are worth further explanation. Computing $\text{Vol} (P)$ involves finding an orthornormal matrix $A$ whose columns span $\mc T_{y}(\mc A)$. This matrix can be found by supplementing $B$ with a set of $n$ linearly independent columns on the right, and applying Gram-Schmidt orthonormalization.  The computational complexity of this step is $\mc O(n^3)$.  This is infeasibly slow when $n$ is large because it must be repeated at each iterate of the MCMC when a complete data set is drawn.  However, using results related to \textit{principal angles} found in \cite{miao1992} the volume \eqref{eq:volume} can be computed using only $B$. $B$ is constructed by Gram-Schmidt orthogonalization of $\nabla s(X,\by), \nabla b_1(X,\by),\dots, \nabla b_p(X,\by)$, reducing the computational complexity to $\mc O(np^2)$--a 
considerable reduction in computational burden when $n \gg p$. 
%Further, the singular values of $P=QA$ are also the singular values of
%$W^\top A$ where $Q=WW^{\top}$, which can be easily obtained through $B$.
The following corollary formally states how computation of $A$ can be circumvented. 
\begin{corollary}
\label{theorem:sings}
Let $U$ be a matrix whose columns form an orthonormal basis for $\mc C (X)$ and set $Q=WW^{\top}$ where the columns of $W$ form an orthonormal basis for $\mc{C}^\perp(X)$. Then the non-unit singular values of $U^\top B$ are the same as the non-unit singular values of $W^\top A$.
\end{corollary} 
\noindent The lemma implies that $\text{Vol} (P)$ is the product of the singular values of $U^\top B$. 

Second, the gradients of $\nabla s(X,\by), \nabla b_1(X,\by),\dots, \nabla b_p(X,\by)$ are easily computed. For example, below we consider M-estimators defined by the estimating equations:
\begin{eqnarray}
\label{Mest}
 \sum_{i=1}^n \psi\left(\frac{y_i - x_{i}^{\top}\bb(\by,X)}{s(\by,X)}\right)= & 0 \\
 \sum_{i=1}^n \chi\left(\frac{y_i - x_{i}^{\top}\bb(\by,X)}{s(\by,X)}\right)= & 0, \nonumber 
\end{eqnarray} 
where $\psi$ and $\chi$ are almost surely differentiable. The gradients can be found by differentiating this system of equations with respect to each $y_{i}$. In theory, finite differences could also be used as an approximation if needed.

\section{Simulated Data}
\label{simData}
We study the performance of restricted likelihood methods in a hierarchical setting where the data are contaminated with outliers. Specifically, simulated data come from the following model:\begin{align}
\label{gensim2}
\begin{split}
& \theta_{i}  \sim   N(\mu, \tau^{2}),  \ i = 1, 2, \dots, 90  \\ 
& y_{ij} \sim (1-p_{i})N(\theta_{i}, \sigma^{2}) + p_{i}N(\theta_{i}, m_{i}\sigma^{2}),\  j = 1, 2,..., n_{i}
\end{split}
\end{align}
with $\mu = 0, \tau^{2} = 1, \sigma^{2} = 4$. The values of $p_{i}, m_{i}$, and $n_{i}$ depend on the group and are formed using 5 replicates of the full factorial design over factors $p_{i},m_{i},n_{i}$ with levels $p_{i} = .1, .2, .3$, $m_{i} = 9, 25$, and $n_{i} = 25, 50, 100$. This results in 90 groups that have varying levels of outlier contamination and sample size. We wish to build models that offer good prediction for the good portion of data within each group. The full model for fitting is a corresponding normal model without contamination:
\begin{equation}
\label{fullsim2}
\begin{split}
%& \mu \propto 1, \  \tau^{2} \propto \tau^{-2}, \\
& \theta_{i}\sim N(\mu, \tau^{2}), \  \sigma^{2}_{i} \sim IG(a_{s}, b_{s}),  \ i = 1, 2, \dots, 90, \\ 
& y_{ij}\sim  N(\theta_{i},\sigma^{2}_{i}), \ j = 1, 2, \dots, n_{i}.
\end{split}
\end{equation}
For the restricted likelihood versions we condition on robust M-estimators of location and scale in each group: $T_{i}(y_{i1}, \dots, y_{in_{i}}) = (\hat\theta_{i}, \hat\sigma^{2}_{i}), i = 1, 2, ..., 90$.  These estimators are solutions to equation \eqref{Mest} (where $x_{i}\equiv 1$) with user specified $\psi$ and $\chi$ functions designed to discount outliers. The two versions use Huber's and Tukey's $\psi$ function, while both versions use Huber's $\chi$ function. The tuning parameters associated with these functions are chosen so that the estimators are $95\%$ efficient under normally distributed data. These classical M-estimators are commonly used in robust regression settings  \citep{huber2009}.  %We will see that his choice has some interesting effects on the results. %These estimators are implemented in the R function \texttt{MASS::rlm} with \cite{huber2009} providing further details. 

To complete the specification of model \eqref{fullsim2}, the hyperparameters $\mu, \tau^{2}, a_{s}$, and $b_{s}$ must be given priors or fixed. The joint prior density for $\mu$ and $\tau^{2}$ is improper and proportional to $\tau^{-2}$. The pair $a_{s}$ and $b_{s}$ are fixed to a variety of values representing different levels of prior knowledge. For each pair, we set $b_{s} = 4a_{s}c$ resulting in a prior mean for each $\sigma^{2}_{i}$ of $\frac{4ca_{s}}{a_{s}-1}, \ a_{s} >1$. The precision is $\frac{(a_{s} -1)^{2}(a_{s}-2)}{(4ca_{s})^{2}}$, meaning larger $a_{s}$ and smaller $c$ result in a more informative prior. With $c = 1$ the shrinkage (for large $a_{s}$) is to the true value of $\sigma^{2} = 4$. We consider $a_{s} = 1.25,  5, 10$ and $c = 0.5, 1, 2$ for a total of nine different priors whose densities are displayed in Figure \ref{fig:priors}. The vertical dashed line is at the known true value of $\sigma^{2} = 4$.

\begin{figure}[t]
\centering
\includegraphics[width = 4in]{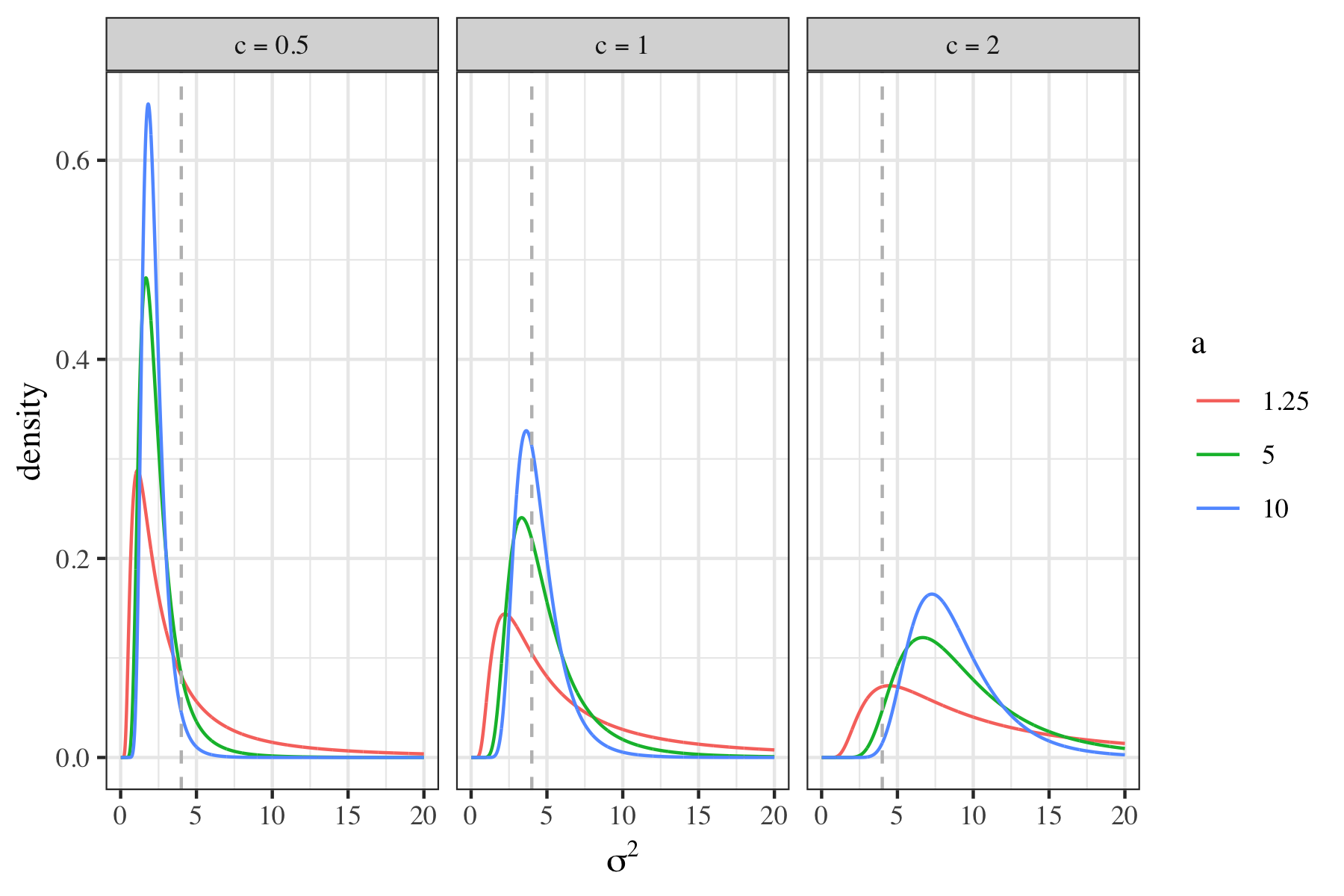}
\caption{The nine different inverse gamma priors used for the group level variance parameters $\sigma_{j}^{2}$ in the simulation. The shape parameter values are $a_{s} = 1.25, 5, 10$. The scale parameter values are $b_{s} = 4a_{s}c$ with $c = 0.5, 1, 2$. The vertical dashed line in each panel is at the true value of $\sigma^{2} = 4$ and is the variance of the good portion of the data in the simulation.}
\label{fig:priors}
\end{figure}

$K = 30$ data sets are generated from \eqref{gensim2}. For each data set and each pair $(a_{s}, c)$, the Bayesian models are fit using MCMC. The MCMC for the restricted likelihood version requires no computational details other than those described for the traditional Bayesian model in Section \ref{BayesLinMod}. This is because there are conditioning statistics for each group and the model's conditional independence between the groups allows the data augmentation described earlier to be performed independently within each group. That is, there is a separate Gibbs step for each group to generate the group level data matching the statistics for that group. 

To assess predictive capability, the models are compared using Kullback-Leibler (KL) divergence from the distribution of good data to the posterior predictive distribution. Specifically, for the $i^{th}$ group of the $k^{th}$ simulated data set $\by_{k}$ compute:
\begin{equation}
\label{kl}
KL^{(M)}_{ik} = \int \log \frac{f(\tilde y | \theta_{i}, \sigma^{2})}{f_{i}(\tilde y | M, \by_{k})}  f(\tilde y | \theta_{i}, \sigma^{2}) \ dy
\end{equation}
where $M$ indexes the fitting model and $f(\tilde y | \theta_{i}, \sigma^{2}) = N(\tilde y | \theta_{i}, \sigma^{2})$, the normal density function with (known) mean $\theta_{i}$ and variance $\sigma^{2}$, evaluated at $\tilde y$. For the Bayesian models, ${f_{i}(\tilde y | M, \by_{k})} = \int f(\tilde y |\theta_{i}, \sigma_{i}^{2}) \pi(\theta_{i}, \sigma^{2}_{i} | M, \by_{k})d\theta_{i}d\sigma^{2}_{i}$ where $\pi(\theta_{i}, \sigma^{2}_{i} | M, \by_{k})$ is the posterior for the $i^{th}$ group model parameters under model $M$ for the $k^{th}$ data set. $M$ denotes either the full normal theory model \eqref{fullsim2} or one of the two restricted likelihood versions, along with specified $a_{s}$ and $c$. For the classical robust fits, we set $f_{i}(\tilde y | M, \by_{k}) = N(\tilde y|\hat\theta_{i}, \hat\sigma^{2}_{i})$ as a groupwise plug-in estimator for the predictive distribution. The classical fits are computed separately for each group with no consideration of the hierarchical structure between the groups. The overall mean $\overline{KL}^{(M)}_{{\cdot}{\cdot}} = \frac{1}{90K} \sum_{k = 1}^{K} \sum_{i=1}^{90} KL^{(M)}_{ik}$ is used to compare the models, where smaller means correspond to better fits. Sampling variation is summarized with the standard error between the $K=30$ replicates in the simulation:  $SE(\overline{KL}^{(M)}_{{\cdot} k}) = \sqrt{\frac{1}{K(K-1)}\sum_{k = 1}^{K} (\overline{KL}^{(M)}_{{\cdot} k} - \overline{KL}^{(M)}_{{\cdot}{\cdot}})^{2}}$ where $\overline{KL}^{(M)}_{{\cdot} k} = \frac{1}{90}\sum_{i = 1}^{90} KL^{(M)}_{ik}$.
 
Figure \ref{kl_sim} displays $\overline{KL}^{(M)}_{{\cdot}{\cdot}}$  with error bars plus/minus one $SE(\overline{KL}^{(M)}_{{\cdot} k})$ for each $a_{s} = 1.25,  5, 10$ and $c = 0.5, 1, 2$. The values of $a_{s}$ and $c$, do not affect the classical robust linear models. The average KL for the normal theory models ranges from $0.22$ to $0.3$ which is much worse than the robust methods and hence is left out of the figure. For $c = 0.5$ and $c = 1$, the results favor the restricted likelihood methods with a slight advantage to the use of Tukey's location estimator over Huber's. This is likely due to the fact that Tukey's estimator essentially trims extreme outliers in the estimation procedure while Huber's estimator discounts them \citep{huber2009}. 
%{\bf John - check the KL for the normal theory models.  .22 to .3 here.  .31 to .37 in the figure caption. }

The choice of $c = 2$ corresponds to a particularly poor prior distribution.  The prior has substantial mass above $\sigma^{2} = 4$, with prior means for $\sigma^2$ from $8.9$ to $32$ as $a_s$ varies.  Additionally, the tuning parameters chosen for the location and scale estimators result in an upward bias in the estimate of $\sigma^{2}$. This bias depends on $m$ and $p$. For example, for $m = 9$ and $p = .1$, Huber's version converges to roughly $4.8$ as $n$ grows.  The bias is greater for more severe levels of contamination.  The alignment of biases in prior distribution and in likelihood from the summary statistic (when applied to the contaminated data) inflates the estimate of scale.  Not surprisingly, a poor prior distribution whose weakness matches the weakness in the likelihood results in poorer inference.  In this case, poorer than the classical estimators. 
 
\begin{figure}[t]
\centering
\includegraphics[width = 4in]{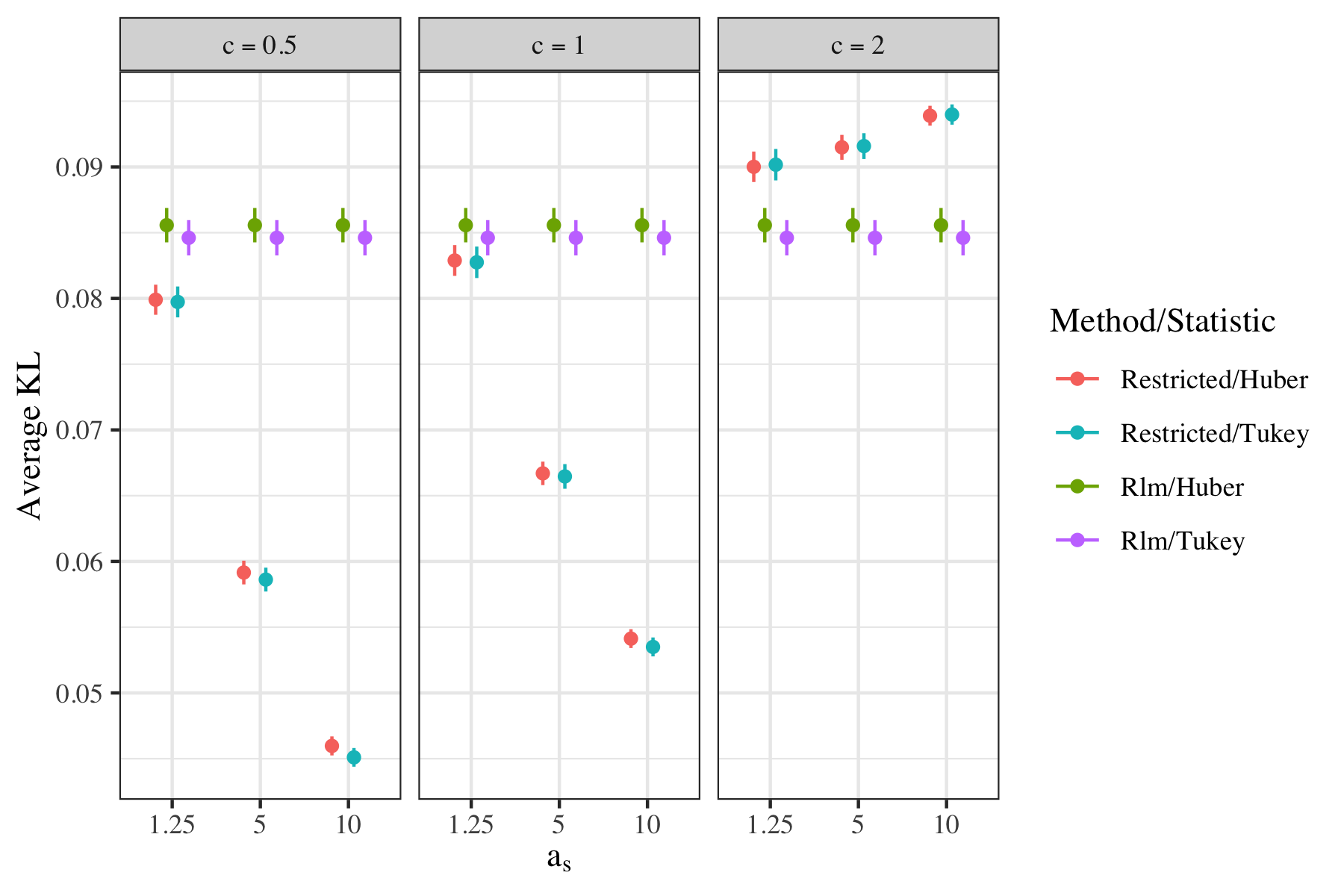}
\caption{Average KL-divergence plus/minus one standard error for each value of $a_{s}$ and $c$ ($\overline{KL}^{(M)}_{{\cdot}{\cdot}} \pm SE(\overline{KL}^{(M)}_{{\cdot} k})$). Smaller values represent better fits. The panels correspond to $c = 0.5$ (left), $c=1$ (middle), and $c=2$ (right), with the values of $a_{s}$ on the horizontal axis. The average KL  for the normal theory model ranges from $0.22$ to $0.3$ and is left out of the figure.}
\label{kl_sim}
\end{figure}

It is also interesting to consider the effects of factors $n$, $p$, and $m$.  We present the results for a single prior ($a_s = 5$ and $c = 1$).  For each simulation $k$, the main effect averages of $KL^{(M)}_{ik}$  are found for each factor $n$, $p$, and $m$. Figure \ref{kl_mnp} displays the average of these main effects over the $K = 30$ simulations along  with error bars plus/minus one standard error. For each group $n$, $p$, and $m$, the Bayesian restricted likelihood versions have better (lower) average KL divergences than do the classical methods. As expected, the average KL gets larger (worse) as the contamination gets more severe (larger $m$ or larger $p$) and the average KL gets smaller (better) as the sample size $n$ grows.  The advantage of the Bayesian method is greater for smaller sample sizes.  

\begin{figure}[t]
\centering
\includegraphics[width = 4in]{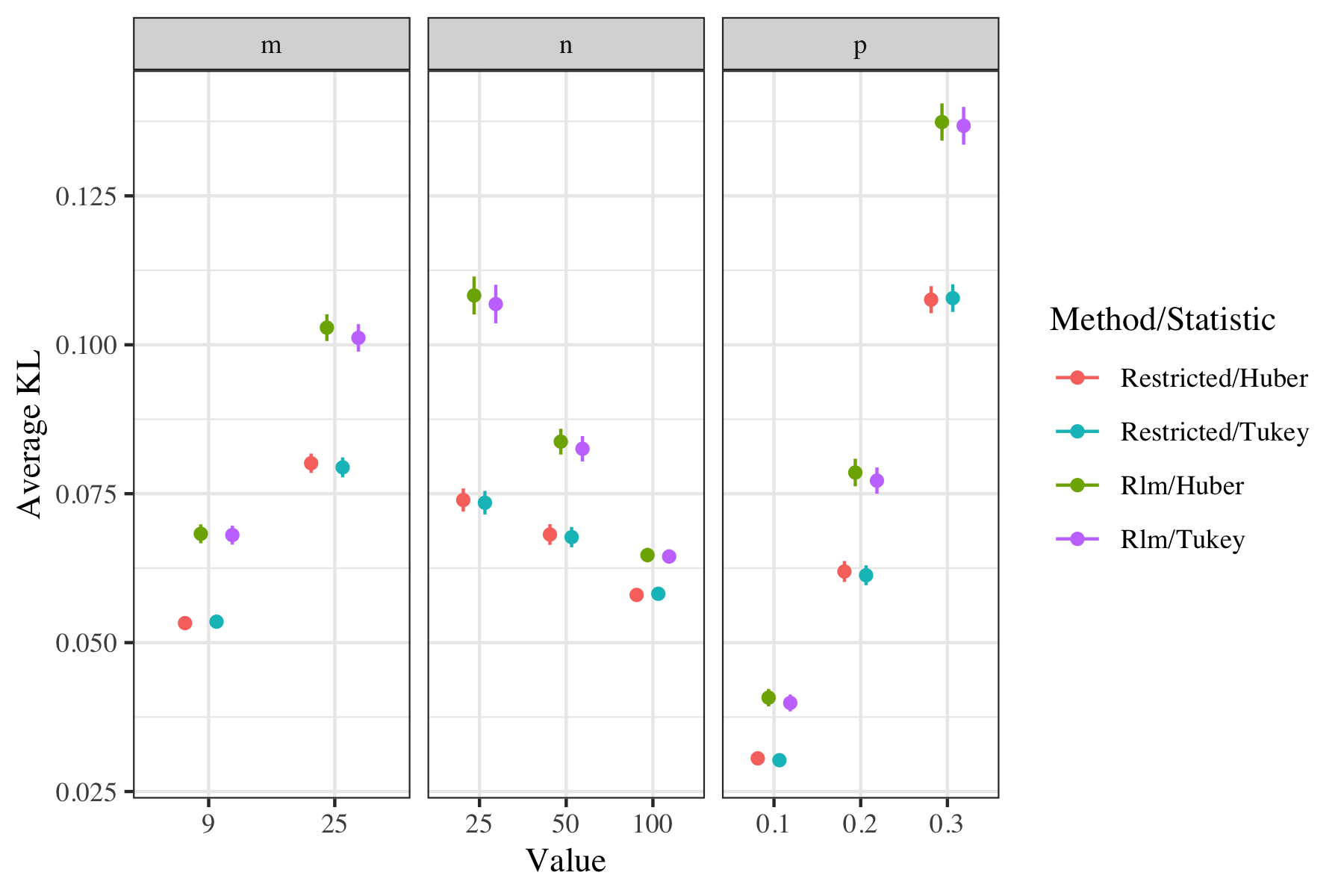}
\caption{Average KL-divergence plus/minus one standard error grouped by the factors $m$ (left), $n$ (middle), and $p$ (right). These results are for the single prior  with $a_s = 5$ and $c = 1$.}
\label{kl_mnp}
\end{figure}

This simulation shows the potential of the restricted likelihood and conveys some cautions.  Specifically, the choice of summary statistics, along with corresponding tuning parameters is important. For the tuning parameters, we applied the default choice of $95\%$ efficiency at the normal.  Under the simulation model here, this choice results in bias in the scale estimation which affects the performance of the method. These choices must be made when using both the classical and Bayesian methods. The Bayesian approach encourages use of a hierarchical model structure and allows one to incorporate prior information in the analysis.  These features can improve predictive performance substantially.  If poorly handled, they can, of course, harm performance.  

%%%%%%%%%%%%%%%%%%%%%%%%%%%%%%%%%%%%%%%%%%%%%%%%%%%%%%%%%%%%
%
% Real data
%
%%%%%%%%%%%%%%%%%%%%%%%%%%%%%%%%%%%%%%%%%%%%%%%%%%%%%%%%%%%%
\section{Real Data}
\label{RealData}
We illustrate our methods with a pair of regression models for data from Nationwide Insurance Company that concern prediction of the performance of insurance agencies.
Nationwide sells many of its insurance policies through agencies which provide direct service to policy holders.  
The contractual agreements between Nationwide and these agencies vary.  Our interest is the prediction of future performance of agencies where  performance is measured by the total number of households an agency services (`household count').  
The data are grouped by states with a varying number of agencies by state. Identifiers such as agency/agent names are removed. Likewise, state labels and agency types (identifying the varying contractual agreements) have been made generic to protect the proprietary nature of the data. Additionally, the counts were scaled to have standard deviation one before analysis. 
As an exploratory view, a plot of the square root of (scaled) household count in 2012, against that in 2010 
is shown in Figure \ref{fig:ctVct} for four states. The states have varying numbers of agencies and  the different colors represent the varying types of contractual agreements as they stood in 2010 (`Type').  A significant number of agencies closed sometime before 2012, as represented by the $0$ counts for 2012. Among the open agencies, linear correlations exists with strength depending on agency type and state.  `Type 1' agencies open in 2012 are of special interest.  One could easily subset the analysis to only these agencies, removing the others. However,  we leave them and use the data as a test bed for our techniques by fitting models that do not account for agency closures or contract type.  Our expectation is that the restricted likelihood will facilitate prediction for the `good' part of the data (i.e., open, `type 1' agencies).  

\begin{figure}[t]
\centering
\includegraphics[width=5in]{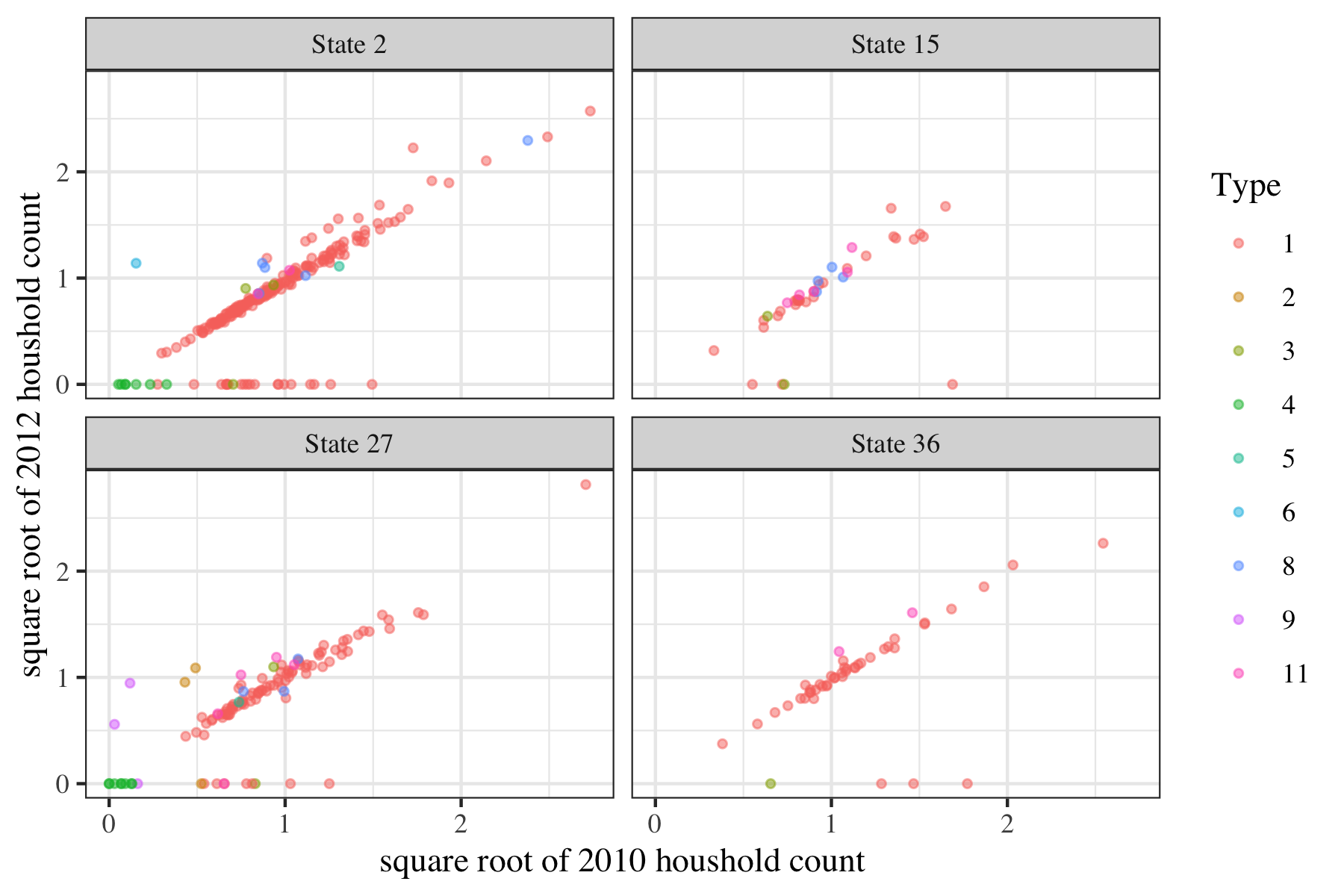}
\caption{The square root of (scaled) count in 2012 versus that in 2010 for four states. The colors represent the varying contractual agreements as they stood in 2010 (`Type').  Agencies that closed during the 2010-2012 period are represented by the zero counts for 2012.}
\label{fig:ctVct}
\end{figure}

\subsection{State Level Regression model}
\label{regModelNW}
The first analysis is based on individual regressions fit separately within states.  The following normal theory regression model is used as the full model for a single state:
\begin{equation}
\label{eq:regModel}
\beta\sim N(\mu_0, \sigma^{2}_0);\ \ \sigma^2\sim IG(a_0,b_0);\ \  
y_{i}=\beta x_{i} +\epsilon_{i},\ \ \epsilon_{i}\iid N(0, \sigma^2),\ i=1,\dots, n, 
\end{equation}
%where $\beta$ is a four dimensional vector ($p=4$) of regression coefficients for the intercept, square root of count in 2010, and the two size/experience measures, and 
where $y_{i}$ and $x_{i}$ are the square rooted household count in 2012 and 2010 for the $i^{th}$ agency, respectively. 
%with covariate vector $\bx_{i}$.  Although the mean of covariates and response have been removed, we include the intercept as fitting is done on a holdout set to evaluate predictive performance.  
The hyper-parameters $a_0, b_0, \mu_0$ and $\sigma^{2}_0$ are all fixed and set from a 
robust regression fit to the corresponding state's data from the time period two years
before. Specifically, Let $\hat\beta$ and $\hat\sigma^{2}$ be estimates from the robust linear regression of 2010 counts on 2008 counts.  We fix $a_0 = 5$ and set $b_0 = \hat\sigma^{2}(a_0 - 1)$ so the prior mean is $\hat\sigma^{2}$. We set $\mu_0 = \hat\beta$ and $\sigma^{2}_0 =  n_{p}se(\hat\beta)^{2}$ where $n_{p}$ is the number of agencies in the prior data set and $se(\hat\beta)$ is the standard error of $\hat\beta$ derived from the robust regression. %The value of $f$ is varied, controlling the prior certainty on $\beta$, with smaller values corresponding to a more certain prior.  
This prior is in the spirit of the Zellner's $g$-prior \citep{zellner1986, liang2008}.  In general, scaling the prior variance $se(\hat\beta)^{2}$ by a factor $g = n_{p}$ is analogous to the unit-information prior \citep{kass1995reference}, with the difference that we are using a prior data set, not the current data set, to set the prior. The obvious reason why this model is misspecified is due to omission of the contract type and agency closure information.  Closing our eyes to these variables, many of the cases appear as outliers. Additionally, the model assumes equal variance within each state, an assumption whose worth is arguable (see Figure \ref{fig:ctVct}). 

We compare four Bayesian models: the standard Bayesian normal theory model, two restricted likelihood models, both with simultaneous M-estimators, and a heavy-tailed model.  For the restricted likelihood methods we use the same simultaneous M-estimators as in the simulation of Section \ref{simData} adapted to linear regression.  The heavy-tailed model replaces the normal sampling density in \eqref{eq:regModel} with a $t$-distribution with $\nu = 5$ degrees of freedom. The Bayesian models are all fit using MCMC, with the restricted versions using the algorithm presented in Section \ref{highDim}. 
We also fit the corresponding classical robust regressions and a least squares regression.  

\subsubsection{Method of model comparison}
We wish to examine the performance of the models in a fashion that preserves the essential features of the 
problem.  Since we are concerned with outliers and model 
misspecification, we understand that our models are imperfect and prefer to use an out-of-sample measure of fit.  
This leads us to cross-validation.  We repeatedly split the data into training
and holdout data sets; fitting the model to the training data and assessing performance on the holdout data.  

The presence of numerous outliers in the data implies that both training and validation data will contain 
outliers.  For this reason, the evaluation must be robust to a certain fraction of bad data.  
The two main strategies are to robustify the evaluation function \citep[e.g.,][]{ronchetti1997} or 
to retain the desired evaluation function and trim cases \citep{jung2014}.  Here,
we pursue the trimming approach with log predictive density for the Bayesian models and log density from plug-in 
maximum likelihood for the classical fits used as the evaluation function.

The trimmed evaluation proceeds as follows in our context.  The evaluation function for case $i$ in the holdout data
is the log predictive density, say
$\log(f(y_i))$, with the conditioning on the 
summary statistic suppressed.  The trimming 
fraction is set at $0 \leq \alpha < 1$. To score a method,
we first identify a base method. Denote the predictive density under this method by $f_{b}(y)$.  Under the base method, $\log(f_{b}(y_i))$ is computed for each case in the 
holdout sample, say $i = 1, \ldots, M$.  Order the holdout sample according to the ordering of $\log(f_{b}(y_i))$ and denote this
ordering by $y_{(1)}^b, y_{(2)}^b, \dots, y_{(M)}^b$. That is, for $i<j$
$\log(f_{b}(y_{(i)}^b))<\log(f_{b}(y_{(j)}^b))$. All of the methods are then scored on the holdout sample with the mean trimmed log marginal pseudo likelihood, 
\[TLM_b(A) = (M - [\alpha M])^{-1} \sum_{i=[\alpha M]+1}^{M}
    \log(f_A(y_{(i)}^b)),\]
 where $f_A$ corresponds to the predictive
 distribution under the method ``A'' being scored.  In other words, the $[\alpha M]$ observations with the smallest values of $\log(f_{b}(y))$ are 
removed from the validation sample and all of the methods are scored using only the
remaining $M - [\alpha M]$ observations. Larger values of $TLM_b(A)$ indicate better predictive performance. This process is advantageous to the base method since the smallest scores from this method are guaranteed to be trimmed.  A method
that performs poorly when it is the base method is discredited.  

\subsubsection{Comparison of predictive performance}
%Model performance is assessed using the mean and standard deviation of the TLM 
%across $K = 50$ different splits into training and holdout samples. 
`Type 1' agencies are of special interest to the company and so the evaluation of the TLM is done on only holdout samples of `Type 1', whereas the training is done on agencies of all types. This is intended to demonstrate the robustness properties of the various methods. 
Models are fit to four states labelled State 2, 15, 27, and 36, with $n = 222, 40, 117,$ and $46$, representing a range of sample sizes. Fitting is done on $K = 50$ training samples with training sample sizes taken to be $0.25n$ and $0.50n$. Holdout evaluation is done on the remaining (`Type 1') samples. For the data augmentation MCMC step under the restricted likelihood models, the acceptance rates range from  $0.16$ to $0.76$ across the states, repetitions, and two versions of the model. The average $TLM_b(A)$ over the $K = 50$ training/holdout samples for the four states and seven methods are shown in Figure \ref{fig:tlm} where the base model is the Student-t model and $\alpha = 0.3$. Similar results are observed for other base models. The error bars are plus/minus one standard deviation of the average $TLM_b(A)$ over the $K = 50$  training/holdout samples. It is clear that the normal Bayesian model used as the full model (Normal) and the classical ordinary least squares fits (OLS) have poor performance due to the significant amount of outlier contamination in the data. In comparing our restricted methods to their corresponding classical methods, there is small, but consistent improvement across the states and training sample size. For state 2, the largest state with $n = 222$, the restricted and classical robust methods have similar performance especially for larger training sample size. This reflects the diminishing effect of the prior as the sample size grows. Notably, the Student-t model performs poorly in comparison for this state. The predictive distribution explicitly accounts for heavy-tailed values, resulting in poorer predictions of the `good' data (i.e., the Type 1 agencies). Likewise, for State 27, another larger state, the Student-t model is outperformed by our restricted methods.   For the other states (State 15 and 36), the Student-t performs similarly to our restricted methods for smaller training sample size (25\% of the sample). However, the performance is slightly worse for the larger training sample size (50\% of the sample). Intuitively, as more data is available for fitting, more outliers appear and the heavy-tailed model compensates for them by assuming they come from the tails of the model; an assumption which is detrimental for prediction. Comparisons of the models depend on $\alpha$ as seen in Figure \ref{fig:tlmbyAlpha} which shows results for different $\alpha$ for training sample size $0.5n$. For smaller $\alpha$ (in this case $\alpha = 0.1$), many outliers are left untrimmed resulting in lower TLM  for all methods and noticeably larger standard deviation for the classical robust methods and our restricted likelihood. Larger values of $\alpha$ ensure that the predictive performance assessment excludes the majority of outliers. The proportion of $0$ counts in the data is roughly $0.14$, suggesting that $\alpha$ should be at least this large.

\begin{figure}[t]
\centering
\includegraphics[width=6in]{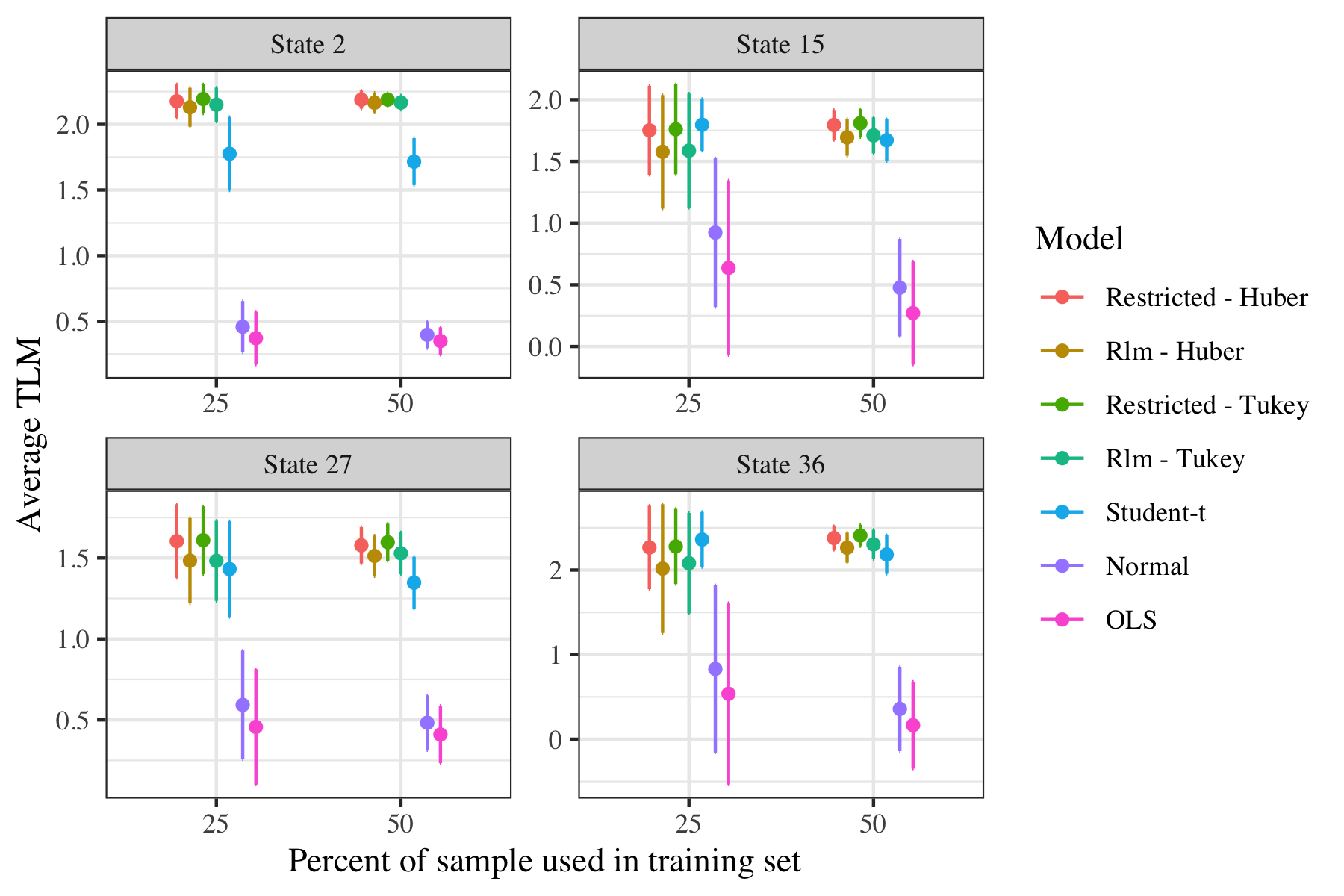}
\caption{Average TLM plus/minus one standard deviation over $K = 50$ splits into training and holdout samples. The panels are for the different states  2, 15, 27, and 36, with $n = 222, 40, 117,$ and $46$, respectively. The horizontal axis is the percent of $n$ used in each training set. The color corresponds to the fitting model.  Larger values of TLM are better.}
\label{fig:tlm}
\end{figure}

\begin{figure}[t]
\centering
\includegraphics[width=6in]{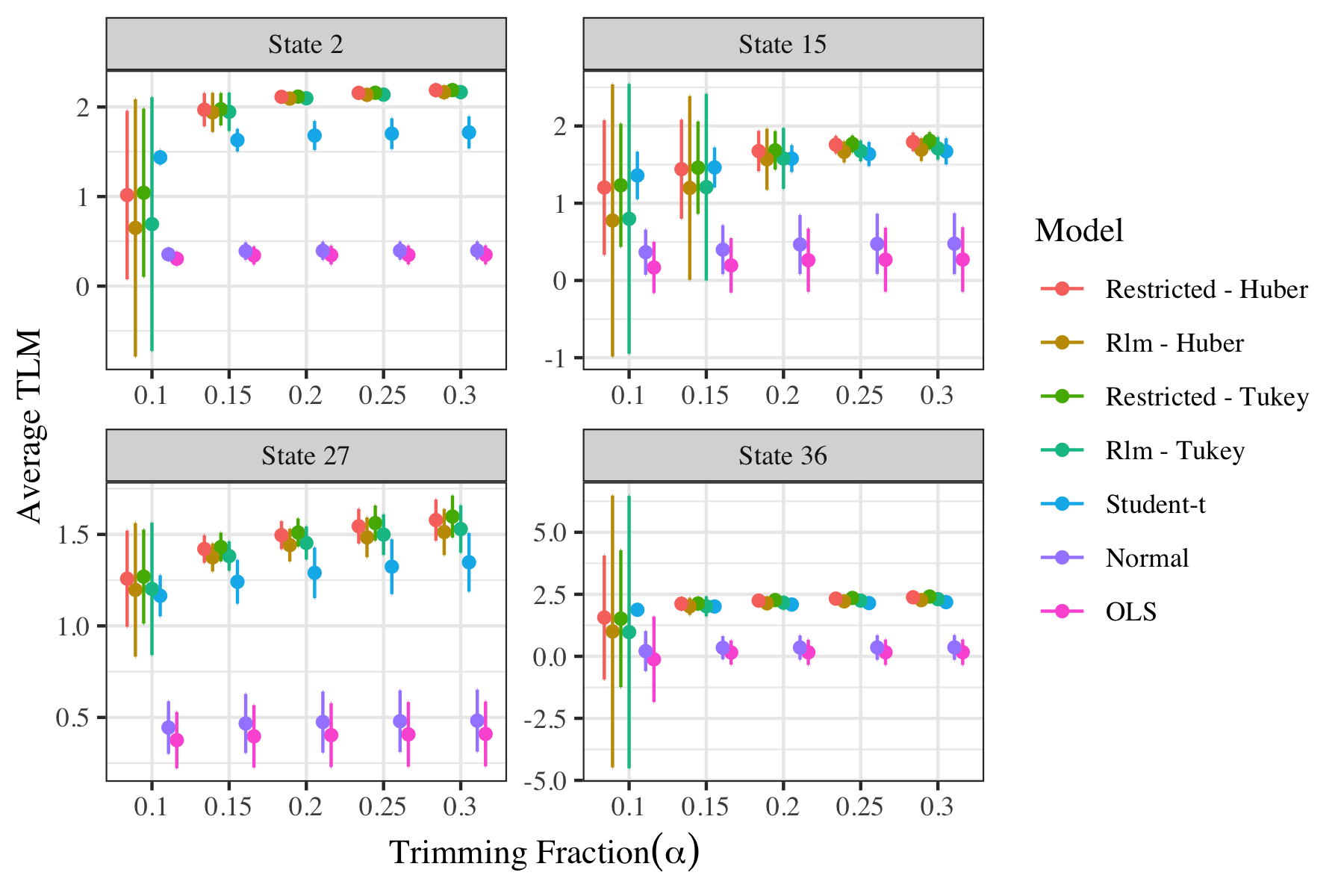}
\caption{Average TLM plus/minus one standard deviation over $K = 50$ splits into training and holdout samples for several values of the trimming fraction $\alpha$. The training sample size used is $0.5n$.  Larger values of TLM are better.}
\label{fig:tlmbyAlpha}
\end{figure}
 
\subsection{Hierarchical regression model}
\label{hierRegNW}
The previous analysis treated states independently. A natural extension is to  reflect similar business environments between states using a hierarchical regression. The proposed model is:
\begin{align}
%\begin{split}
\label{eq:hierModel}
&\beta\sim N_p(\mu_0, a\sigma_0^{2});\ \ 
\beta_j\iid N_p(\beta, b\sigma_0^{2}); \ \  
\sigma_j^2\sim IG(a_0,b_0);  & \\ \nonumber
& y_{ij}=x_{ij}\beta_j+\epsilon_{ij},\ \ \epsilon_{ij}\iid N(0, \sigma_j^2),\ i=1,\dots, n_j,\ j=1,\dots, J &
%\end{split}
\end{align}
where $y_{ij}$ is the $i^{th}$ observation of square rooted household count in 2012 in the $j^{th}$
state, $n_{j}$ is the total number of agencies in state $j$, and $J$ is
the number of states. $x_{ij}$ is the square rooted household count in 2010 and $\beta_j$ represents the individual regression coefficient vector for state $j$. The parameters $\mu_0$,
$\sigma^{2}_0$, $a_0$, and $b_0$ are fixed by fitting  the regression $y_{ij}=x_{ij}\beta+\epsilon_{ij}$ using Huber's M-estimators to the prior data set from two years before. Using the estimates from this model, we set $\mu_{0} = \hat\beta$, $\sigma_{0}^{2} = n_{p}se(\hat\beta)^{2}$ ($n_{p} = 2996$ is the number of observations in the prior data set), $a_{0}=5$ and $b_{0} = \hat\sigma^{2}(a_{0} -1)$. We constrain $a+b=1$
in an attempt to partition the total variance between the individual
$\beta_j$'s and the overall $\beta$. We take $b\sim
\text{beta}(v_1,v_2)$. Using the prior data set, we assess the
variation between individual estimates of the $\beta_j$ to set $v_1$
and $v_2$ to allow for a reasonable amount of shrinkage. To allow for
dependence across the $\sigma_j^2$ we first take
$(z_1,\dots,z_J)\sim N_J(\mathbf{0}, \Sigma_\rho)$ with
$\Sigma_\rho=(1-\rho)\mb{I}+\rho \mb{1}\mb{1}^{\top}$. Then we set
$\sigma^2_j=H^{-1}(\Phi(z_j))$ where $H$ is the cdf of an
$IG(a_0,b_0)$ and $\Phi$ is the cdf of a standard normal. This results in the specified marginal distribution, while
introducing correlation via $\rho$. We assume $\rho\sim
\text{beta}(a_\rho,b_\rho)$ with mean $\mu_\rho=a_\rho/(a_\rho+b_\rho)$ and precision
$\psi_\rho=a_\rho+b_\rho$. The parameters $\mu_\rho$ and
$\psi_{\rho}$  are given beta and gamma distributions, with fixed hyperparameters. More details on setting prior parameters are given in the appendix. 

Using the same techniques as in the previous section, 
we fit the normal theory hierarchical model above, a thick-tailed $t$ version with $\nu = 5$ d.f., and two restricted likelihood versions (Huber's and Tukey's) of the model.  For the restricted methods, we condition on robust regression estimates fit separately within each state. We also fit classical robust regression counterparts and a least squares regression separately within each state. Hierarchical models naturally require more
data and so we include states having at least 25 agencies resulting in 22 states in total and $n = \sum_{j} n_{j} =  3180$ total agencies. For training data we take a stratified (by state) sample of size $3180/2 = 1590$ where the strata sizes are $n_{j}/2$ (rounded to the nearest integer). The remaining data is used for a holdout evaluation using TLM computed separately within each state: $TLM_b(A)_{j} = (M_{j} - [\alpha M_{j}])^{-1} \sum_{i=[\alpha M_{j}]+1}^{M_{j}} \log(f_A(y_{(i)j}^b))$ where $y_{(1)j}^b, y_{(2)j}^b,..., y_{(M_{j})j}^b$ is the ordering of the $M_{j}$ holdout observations within state $j$ according to the log marginals under the base model $b$. For the non-Bayesian models,  $f_A(y^b_{(i)j})$ is estimated using plug-in estimators for the parameters for state $j$. $TLM_b(A)_{j}$ is computed for each state for $K=50$ splits of training and holdout sets. The Bayesian models are fit using MCMC, with the restricted versions applying the algorithm laid out in Section \ref{BayesLinMod} and adapted to the hierarchical setting as described in Section \ref{simData}. For the MH-step proposing augmented data,  the acceptance rates for the two restricted likelihood models across all states and repetitions range from $0.24$ to  $0.74$.

The average over states, $\overline{TLM}_b(A)_{\cdot}= \frac{1}{22} \sum_{j =1}^{22} TLM_b(A)_{j}$ for each of the $K$ repetitions is summarized in Figure
\ref{fig:hierTLM} for several trimming fractions using the Student-t as the base model. The points are the average of the $\overline{TLM}_b(A)_{\cdot}$ over the $K$ repetitions with error bars plus/minus one standard deviation over $K$ with larger values representing better predictive performance. As the trimming fraction used for the TLM increases, so does TLM since more outliers are being trimmed. Similar patterns were seen in the individual state level regressions in Section \ref{regModelNW}. Despite being used as the base model to compute TLM, the Student-t doesn't perform well in comparison to the robust regressions. We attribute this to the assumption of heavier tails resulting in smaller log marginal values on average; emphasizing again that the t-model will do well to discount outlying observations but does not provide a natural mechanism for predicting `good' (i.e., non-outlying) data. For each trimming fraction, our restricted likelihood hierarchical models outperform the classical robust regressions fit separately within each state. The hierarchical model also reduces variance in predictions resulting in smaller error bars. This improvement decreases with $\alpha$ but is still noticeable for $\alpha = 0.2$. Both the Tukey and Huber versions perform similarly. 

\begin{figure}[t]
\centering
\includegraphics[width=6in]{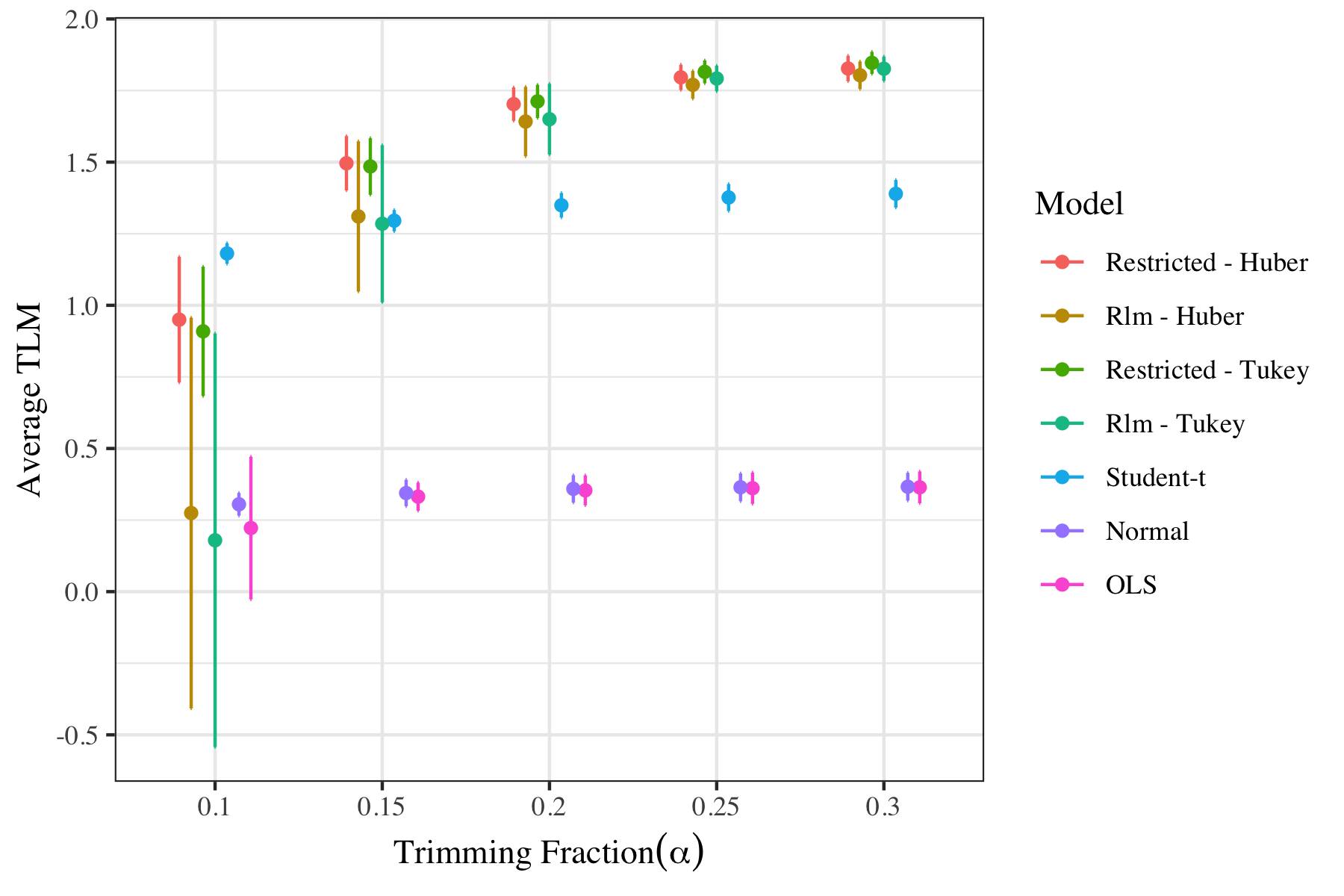}
\caption{Hierarchical model results: $\overline{TLM}_b(A)_{\cdot}$  plus/minus one standard deviation over $K = 50$ splits into training and holdout sets with the Student-t as the base model and several values of the trimming fraction $\alpha$. Larger values of TLM are better.}
\label{fig:hierTLM}
\end{figure}

It is also interesting to examine the results within each state. Figure \ref{fig:hierTLMstate} summarizes ${TLM}_b(A)_{j}$ with $\alpha = 0.3$ for each state where the points and error bars are the averages and plus/minus one standard deviation of ${TLM}_b(A)_{j}$ over the $K = 50$ repetitions. The results are only given for the models using Tukey's M-estimators (Huber's version looks similar). The states are ordered along the $x$-axis according to number of agencies within the state (shown in parentheses). In several of the smaller states, the restricted hierarchical model performs better with similar performance between the models in most of the larger states, a reflection of the decreased influence of the prior.  The hierarchical structure pools information across states, improving performance in the smaller states. The standard deviations are smaller for the hierarchical model in smaller states than they are for the corresponding classical model.  In larger states, the standard deviations are virtually identical. Similar benefits are often seen for hierarchical models \citep[e.g.,][]{gelman2006}.

\begin{figure}[t]
\centering
\includegraphics[width=6.8in]{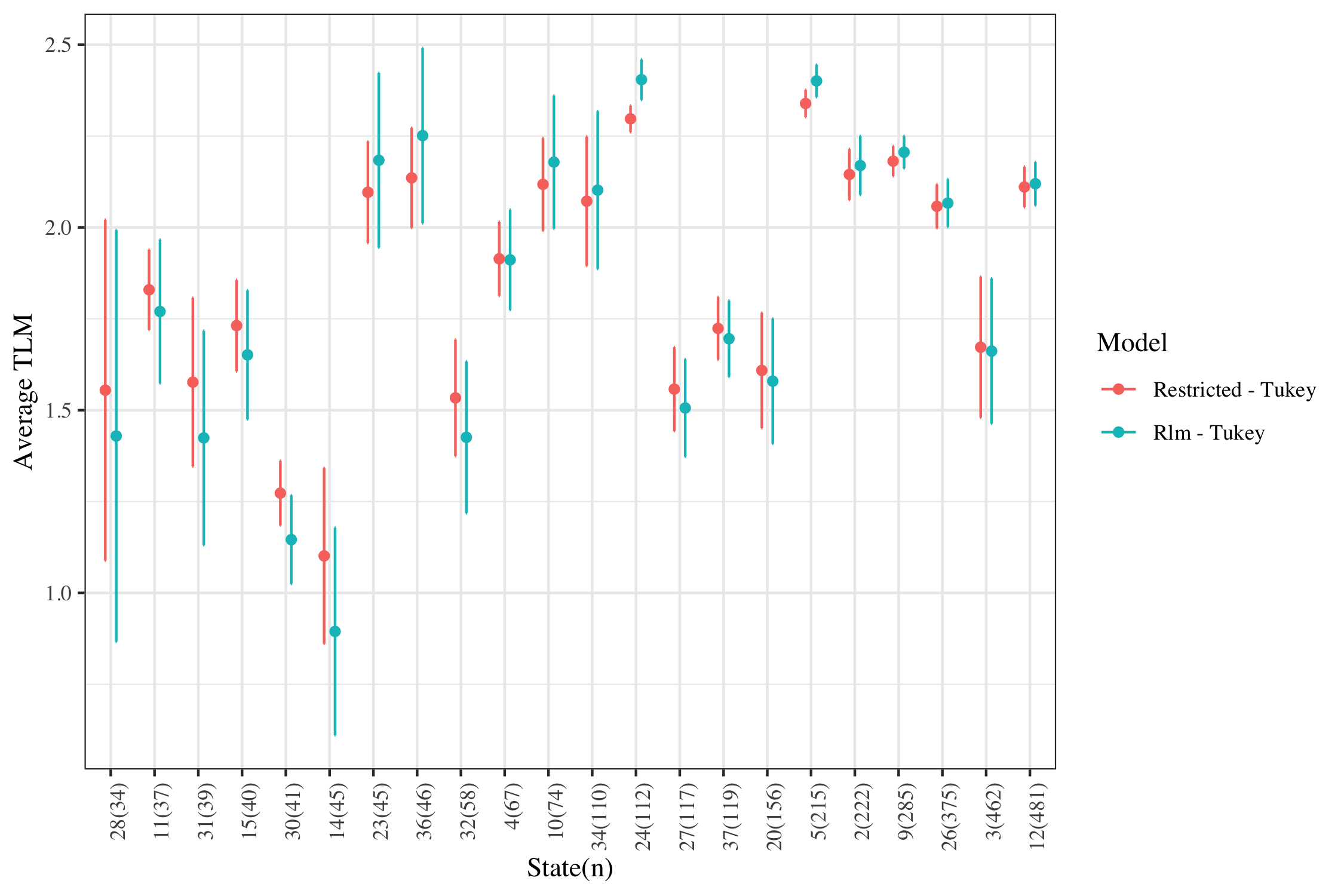}
\caption{Hierarchical model results: ${TLM}_b(A)_{j}$  plus/minus one standard deviation over $K = 50$ repetitions for each state and $\alpha = 0.3$. The states are ordered along the $x$-axis according to number of agencies within the state (shown in parentheses). Results displayed are for the robust models using Tukey's M-estimators. Larger values of TLM are better.}
\label{fig:hierTLMstate}
\end{figure}

\section{Discussion}
\label{Conclusions}

This paper develops a Bayesian version of restricted likelihood where posterior inference is conducted by conditioning on a summary statistic rather than the complete data.  The framework blends classical estimation with Bayesian methods.  
%Many routine choices in an analysis react to the gap between reality and the
%statistical model, where a bit of set-up work improves inferential
%performance.  Often, these choices can be recast in the framework of
%restricted likelihood presented here, lending them more formality and facilitating
%development of theoretical results. 
%A greater benefit of our framework is that it leads us to blend classical estimation with Bayesian methods. 
Here, we concentrate on outlier-prone settings where natural choices for the conditioning statistic are classical robust estimators targeting the mean of the non-outlying data (e.g., M-estimators).  The likelihood conditioned on these estimators is used to move from prior to posterior. The update follows Bayes' Theorem, conditioning on the observed estimators exactly.   Computation is driven by MCMC methods, requiring only a supplement to existing algorithms by adding a Gibbs step to sample from the space of data sets satisfying the observed statistic. This step has additional computation costs arising from the need to compute the estimator and an orthonormal basis derived from gradients of the estimator at each iteration. The cost of finding the basis can be reduced by exploiting properties of the geometric space from which the samples are drawn as described in Section~\ref{highDim}. We have seen good mixing of the MCMC chains across a wide-variety of examples.

The Bayesian restricted likelihood framework can be used to address model misspecification, of which the presence of outliers is but one example. The traditional view is that, if the model is inadequate, one should build a better model. In our empirical work, as data sets have become larger and more complex, we have bumped into settings where we cannot realistically build the perfect model. We ask the question ``by attempting to improve our model through elaboration, will the overall performance of the model suffer?'' If yes, we avoid the elaboration, retaining a model with some level of misspecification. Acknowledging that the model is misspecified implies acknowledging that the sampling density is incorrect, exactly as we do when outliers are present. In this sense, misspecified models and outliers are reflections of the same phenomenon, and we see restricted likelihood as a method for dealing with this more general problem. 

Outside of outlier-prone settings, we might condition on the results of a set of estimating equations designed to enforce a lexical preference for those features of the analysis considered most important, yet still producing inferences for secondary aspects of the problem. This leads to questions regarding the choice of summary statistic to apply. In the literature, great ingenuity has been used to create a wide variety of estimators designed to handle specific manifestations of a misspecified model.  The estimators are typically accompanied by asymptotic results on consistency and limiting distribution.  These results can be used as a starting point to choose appropriate conditioning statistics in specific settings.  For example, a set of regression quantiles may be judged the most important feature of a model.  It would then be natural to condition on the estimated regression quantiles and to use a flexible prior distribution to allow for nonlinearities in the quantiles.  The computational strategies we have devised allow us to apply our methods in this setting and to make full predictive inference.  In general, we recommend a choice of conditioning statistic based on the analyst's understanding of the problem, model, reality, deficiencies in the model,  inferences to be made, and the relative importance of various inferences.  

The framework we develop here allows us to retain many benefits of Bayesian methods:  it requires a complete model for the data; it lets us combine various sources of information both through the use of a prior distribution and through creation of a hierarchical model; it guarantees admissibility of our decision rules among the class based on the summary statistic $T(\by)$; and it naturally leads us to focus on predictive inference.  The work does open a number of questions for further work, including a need to investigate restricted likelihood methods as they relate to model selection, model averaging for predictive performance, and model diagnostics.

\section{Appendix}
\label{sec:appendix}
\subsection{Proofs}
\noindent

Proof of Theorem~\ref{Transformation}.  
\begin{proof} 
\begin{eqnarray}
 s(X,\by) & = & s\left(X,\frac{s(X,\by_{obs})}{s(X,\bz^*)}\bz^* + X\left(\bb(X,\by_{obs}) - \bb(X,\frac{s(X,\by_{obs})}{s(X,\bz^*)}\bz^*)\right)\right) \\
& = & \frac{s(X,\by_{obs})}{s(X,\bz^*)} s(X, \bz^*)= s(X,\by_{obs}) , \qquad \mbox{and} \\
 \bb(X,\by) & = & \bb\left(X,\frac{s(X,\by_{obs})}{s(X,\bz^*)}\bz^* + X\left(\bb(X,\by_{obs}) - \bb(X,\frac{s(X,\by_{obs})}{s(X,\bz^*)}\bz^*)\right)\right) \\
 & = & \bb(X,\frac{s(X,\by_{obs})}{s(X,\bz^*)}\bz^*) + \bb(X,\by_{obs}) - \bb(X,\frac{s(X,\by_{obs})}{s(X,\bz^*)}\bz^*) \\ &=& \bb(X,\by_{obs})
\end{eqnarray}
\end{proof}

\noindent
\begin{theorem}
\label{1to1onto}
The mapping $h:  \mathbb{S} \rightarrow \mathcal{A}$ with $h$ defined in Theorem \ref{Transformation} is one-to-one and onto. 
\end{theorem}
\begin{proof} 
\noindent 
\textit{One-to-one}: Let $z_{1}, z_{2} \in \mathbb{S}$ with $h(z_{1}) = h(z_{2})$. Rearrangement implies $z_{1} = cz_{2} + Xv$ for known $c\in \mathbb{R}$ and $v\in \mathbb{R}^{p}$ depending on $\bb(X,\by_{obs})$, $s(X,\by_{obs})$, $\bb(X,z_{1})$, $s(X,z_{1})$, $\bb(X,z_{2})$, $s(X,z_{2})$. Given $z_{2}\in \mathbb{S}$, $v\neq 0$ implies $z_{1}\notin \mathcal{C}^{\perp}(X)$ and $c\neq 1$ implies $||z_{1}|| \neq 1$. Thus $z_{1} \in \mathbb{S}$ implies $c=1$ and $v =0$. 

\noindent \textit{Onto:} Let $\by \in \mathcal{A}$ and consider its projection onto $\mathcal{C}^{\perp}(X)$: $Q\by$ where $Q = I - XX^{\top}$. It is easy to show that $\bz^{*} = Q\by/||Q\by|| \in \mathbb{S}$ and $h(\bz^{*}) = \by$.
\end{proof}

\noindent
Proof of Lemma~\ref{gradSTheoremReg}.
\begin{proof}
We first show that $\nabla s(X,\by)\in \mc{C}^\perp(X)$. Recall that
$H=I-Q$. By the regression invariance property \ref{regIn}, we have
\label{perpGradReg}
\begin{equation}
\label{eq:lem3.2}
\begin{aligned}
s(X,\by)=s(X, Q\by+H\by)=s(X, Q\by).
\end{aligned}
\end{equation}
Thus, by the chain rule $\nabla s(X,\by)=Q\nabla s(X,Q\by)=Q\nabla s(X, \bz)$. Hence $X^\top \nabla s(X,\by)=0$ as desired.
From equation~\eqref{eq:lem3.2}, all vectors $\bz'\in \Pi(\mathcal{A})$ satisfy $s(X,\bz')=
s(X,\by)=s(X,\by_{obs})$, and so all directional derivatives of $s$ along each tangent $\bv$ to
  $\Pi(\mathcal{A})$ in $\mc C^\perp(X)$ at $\bz$ are equal to 0 (i.e., $\nabla s(X,\bz) \cdot \bv=0$).  Thus $\nabla s(X,\bz)$ is orthogonal to  $\Pi(\mathcal{A})$ at $\bz$.  
Since $\Pi(\mathcal{A})$ has dimension $n-p-1$, $\nabla s(X,\bz)$ gives the unique (up to scaling and reversing direction) normal in the $n-p$ dimensional $\mc C^\perp(X)$.  
\end{proof}

\noindent
Proof of Lemma~\ref{lem:basis}

\begin{proof}
Without loss of generality, assume the columns of $X$ form an
orthonormal basis for $\mc C (X)$ and likewise the columns of $W$ form
and orthonormal basis for $\mc C^\perp(X)$. With earlier notation,
$H=XX^{\top}$ and $Q=WW^{\top}$. The set $\mc A$ is defined by the
$p+1$ equations  $s(X,\by)=s(X,\by_{obs})$, 
$b_1(X,\by)=b_1(X,\by_{obs}),\dots,  b_p(X,\by)=b_p(X,\by_{obs})$. Consequently, the gradients are orthogonal to $\mc A$. Let  $\nabla\bb(X,\by)$ denote the $n\times p$ matrix with columns $\nabla b_1(X,\by),\dots, \nabla b_p(X,\by)$. We seek to show the $n \times (p+1)$ matrix $[\nabla\boldsymbol\bb(X,\by),\nabla s(X,\by)]$ has rank $p+1$. Using property \ref{regEq}, we have that 
\[
\bb(X, \by)=\bb(X,Q\by+H\by)=\bb(X, Q\by)+X^\top \by
\] 
Then $\nabla \bb(X,\by)=Q\nabla\boldsymbol\bb(X, Q\by)+ X$ and 
\begin{eqnarray}
\label{BigMatrix}
[XX^\top, WW^\top]^\top[\nabla\boldsymbol\bb(X,\by),\nabla s(X,\by)]=
 \left( \begin{array}{cc}
X & \mathbf{0} \\
WW^\top\nabla b(X,\by)  &\nabla s(X,\by)  \\ \end{array} \right)
\end{eqnarray}
The last column comes from Lemma \ref{gradSTheoremReg}. The matrix $[XX^\top, WW^\top]^\top$ is of full
column rank (rank $n$), and so the rank of $[\nabla\boldsymbol\bb(X,\by),\nabla s(X,\by)]$ is the same as the rank
of the matrix on the right hand side of (\ref{BigMatrix}).  This last
matrix has rank $p+1$ since $\nabla s(X,\by) \ne \bzero$ by \ref{scaleEq2Reg}, and so does 
$[\nabla b(X,\by),\nabla s(X,\by)]$.
\end{proof}

\noindent
Proof of Lemma~\ref{lem:fullrank}

\begin{proof}
$P$ is the projection of the columns of $A$ onto $\mc
C^{\perp}(X)$. For this to result in a loss of rank, a subspace of
$\mc T_{y}(\mc A)$ must belong to $\mc C(X)$.  Following property
\ref{regEq}, for an arbitrary vector $X \bv \in \mc C(X)$, $\bb(X,\by
+ X \bv) = \bb(X,\by) + \bv$.  From the property, we can show that the directional derivative
  of $\bb$ along $X \bv$ with $\bv \ne \bzero$ is $\bv$, which is a
  nonzero vector. Hence $X\bv \notin \mc T_{y}(\mc A)$.  
\end{proof}

\noindent
Proof of Corollary~\ref{theorem:sings}

\begin{proof}
The corollary relies on a lemma and theorem from \cite{miao1992} which we restate 
slightly for brevity of presentation.  The principal angles between subspaces pluck off a
set of angles between subspaces, from smallest to largest.  The number of such angles 
is the minimum of the dimensions of the two subspaces.  Miao and Ben-Israel's first result
(their Lemma 1) connects these principal angles to a set of singular values, and hence to 
volumes.   
\begin{lemma}{(Miao, Ben-Israel)}
\label{MBI:lemma}
Let the columns of $Q_L\in \mathbb{R}^{n\times l}$ and $Q_M\in
\mathbb{R}^{n\times m}$ form orthonormal bases for linear subspaces
$L$ and $M$ respectively, with $l \leq m$. Let $\sigma_1\geq\cdots\geq
\sigma_l\geq0$ be the singular values of $Q_M^\top Q_L$. Then $\cos
\theta_i=\sigma_i, i=1,\dots,l$ where $0\leq\theta_1\leq\theta_2\leq
\cdots \leq\theta_l\leq\frac{\pi}{2}$ are the principal angles between $L$ and $M$.  
\end{lemma}

Miao and Ben-Israel's second result (their Theorem 3) makes a match between the principal
angles between a pair of subspaces and the principal angles between their orthogonal complements.  
\begin{theorem}{(Miao, Ben-Israel)}
\label{MBI:thm}
The nonzero principal angles between subspace $L$ and $M$ are equal to the 
nonzero principal angles between $L^\perp$ and $M^\perp$.
\end{theorem}

To establish the corollary, we appeal to Lemma~\ref{MBI:lemma} and Theorem~\ref{MBI:thm}.  Translating Miao and Ben Israel's
notation, we have $M=\mc C^\perp (X)$, $Q_M=W$, $L=\mc
T_{\boldsymbol{y}}(\mc{A})$, and $Q_L= A$. By Theorem~\ref{MBI:thm}, the
nonzero principal angles between $\mc{T}_{\boldsymbol{y}}(\mc{A})$ and
$\mc C^\perp(X)$ are the same as the nonzero principal angles between
$\mathcal{T}_{\boldsymbol{y}}^\perp(\mathcal{A})$ and $\mc C(X)$. By
\ref{MBI:lemma}, the non-unit singular values of $W^\top A$ are the
same as the non-unit singular values of $U^\top B$.  
\end{proof}

\subsection{Setting the hierarchical prior values}
 
This section describes the how the prior parameters are set in  Section \ref{hierRegNW}. Using the previous data set from two years prior, we fit separate (robust) regressions to each state and a  regression to the entirety of the data at once. Let the estimates for the fits to each state be $\hat{\beta_{1}}, \dots, \hat \beta_{J}, \hat \sigma_{1}, \dots, \hat \sigma_{J}$ and the estimates from the single regression be $\hat \beta$ and $\hat \sigma$. These are classical robust estimates using Tukey's regression and Huber's scale. For this sections, let $n_{j}$ denote the number of observations in the $j^{th}$ state (of the previous data set) and set $n_{p}=\sum n_{j}$. 

First, consider $v_{1}$ and $v_{2}$ in the prior $b\sim\text{beta}(v_{1},v_{2})$.  In the hierarchical model \eqref{eq:hierModel}, $b=0$ implies all the $\beta_{j}'s$ are equal (no variation between states) and $b=1$ implies the $\beta_{j}'s$ vary about $\mu_{0}$ according to $\Sigma_{0}=n_{p}\cdot \mbox{se}(\hat\beta)^{2}$ (see Section \ref{regModelNW}). We seek a prior measure for what we think $b$ should be. Using the prior fit, a measure for  uncertainty for $\beta$ is $\Sigma_{\hat\beta}=\mbox{se}(\hat\beta)^{2}$, the estimate of the variance from the single regression. For the $\beta_{j}'s$, take $\delta_{j}=\hat\beta_{j}-\hat\beta$ and set the prior uncertainty to $\Sigma_{\delta}=n_{p}^{-1}\sum_{j} n_{j}\delta_{j}^{2}$. Consider  $g= \Sigma_{\delta}/\Sigma_{\hat\beta}$ measuring of the amount of uncertainty between the $\beta_{j}'s$ relative to that of $\beta$. Now in the prior, we heuristically set the uncertainty in the $\beta_{j}'s$ ($b\Sigma_{0}$) to be approximately equal to $g\cdot\Sigma_{\hat\beta}$. That is, $b\Sigma_{0}\approx g\cdot\Sigma_{\hat\beta}= \frac{g}{n} \Sigma_{0}$, suggesting $b\approx  \frac{g}{n}$.  Thus, we set $E[b]=\frac{g}{n}$. The precision, $v_{1}+v_{2}$, is set to $10$, completing the specification for the prior on $b$. 

Finally, recall $\rho\sim\text{beta}(a_\rho,b_\rho)$ with mean $\mu_\rho=a_\rho/(a_\rho+b_\rho)$ given a beta prior and precision
$\psi_\rho=a_\rho+b_\rho$ given a gamma prior. There is little evidence of any strong correlation amongst estimates of $\sigma^{2}_{j}$ in the prior data set and we set the prior mean of $\mu_{\rho}$ equal to $0.2$ and prior variance to $.01$. Noting $\text{var}(\rho|\mu_{\rho}, \psi_{\rho})=\mu_{\rho} (1-\mu_{p})/(\psi_{\rho}+1)$ we plug in $\mu_{\rho} = 0.2$ and $\text{var}(\rho|\mu_{\rho}, \psi_{\rho}) = 0.01$. Solving for $\psi_{\rho}$ results in a value of $15$. This is taken to be the mean of the gamma prior on $\psi_{\rho}$. Finally, we  set the rate parameter for to 1 implying the variance of the gamma prior is equal to its the mean. With this specification, the prior on $\rho$ has 80\% of the central mass between roughly $0.03$ and $0.42$ and reflects our prior belief that there is likely only weak positive correlation amongst the $\sigma^{2}_{j}$'s.
%%The variance was set to twice the inverse Fisher information evaluated at $\hat\rho_{mle}$

\bibliographystyle{apa}
\bibliography{refPaper1}

\end{document}